\documentclass[a4paper, 11pt]{article}
\usepackage[utf8]{inputenc}
\usepackage[ngermanb, english]{babel}
\usepackage{amsmath}
\usepackage{amssymb}
\usepackage{amsthm}
\usepackage{euscript}
\usepackage{tikz-cd}
\usepackage{bm}
\usepackage{commath}
\usepackage{mathrsfs}
\usepackage{mathtools}
\usepackage{slashed}
\usepackage{tensor}
\usepackage{parskip}
\usepackage{epstopdf}
\usepackage{graphicx}
\usepackage{hhline}
\usepackage[nottoc]{tocbibind}
\usepackage{babelbib}
\usepackage{hyperref}
\usepackage[top=2.5cm, left=2.5cm, right=2.5cm, bottom=2.5cm]{geometry}

\theoremstyle{plain}
\newtheorem{thm}{Theorem}[section]
\newtheorem{lem}[thm]{Lemma}
\newtheorem{prop}[thm]{Proposition}

\newtheorem{col}[thm]{Corollary}

\theoremstyle{definition}
\newtheorem{defn}[thm]{Definition}

\newtheorem{ass}[thm]{Assumption}

\newtheorem{exmp}[thm]{Example}

\theoremstyle{remark}
\newtheorem{rem}[thm]{Remark}

\newcommand{\Q}{{\boldsymbol{\mathrm{Q}}}}
\newcommand{\BQ}{\mathcal{B}_\Q}
\newcommand{\CQ}{{\mathcal{C}_\Q}}
\newcommand{\FQ}{\mathcal{F}_\Q}
\newcommand{\dt}[1]{\operatorname{Det} \left ( #1 \right )}

\newcommand{\textfrac}[2]{#1 / #2}
\newcommand{\imaginary}{\mathrm{i}}

\newcommand{\enter}{\vspace{\baselineskip}}
\newcommand{\sctn}[1]{\Gamma \left ( #1 \right )}
\newcommand{\sctnbig}[1]{\Gamma \bigl ( #1 \bigr )}

\newcommand{\vectc}[1]{\mathfrak{X}_\text{c} \left ( #1 \right )}

\newcommand{\deDonder}{{d \negmedspace D \mspace{-2mu}}}
\newcommand{\linLorenz}{{l \negmedspace L \mspace{-2mu}}}
\newcommand{\Lorenz}{L}
\newcommand{\diff}{\mathfrak{diff} \left ( M \right )}
\newcommand{\Diff}{\operatorname{Diff}_0 \left ( M \right )}
\newcommand{\Lie}{\mathsterling}
\newcommand{\lie}{\ell}

\newcommand{\commutatorbig}[2]{\bigl [ #1 , #2 \bigr ]}
\newcommand{\setbig}[1]{\bigl \{ #1 \bigr \}}

\newsavebox{\foobox}

\makeatletter
\newcommand{\subalign}[1]{%
  \vcenter{%
    \Let@ \restore@math@cr \default@tag
    \baselineskip\fontdimen10 \scriptfont\tw@
    \advance\baselineskip\fontdimen12 \scriptfont\tw@
    \lineskip\thr@@\fontdimen8 \scriptfont\thr@@
    \lineskiplimit\lineskip
    \ialign{\hfil$\m@th\scriptstyle##$&$\m@th\scriptstyle{}##$\crcr
      #1\crcr
    }%
  }
}
\makeatother

\providecommand{\sectionref}[1]{Section~\ref{#1}}
\providecommand{\sectionsaref}[2]{Sections~\ref{#1} and \ref{#2}}

\providecommand{\eqnref}[1]{Equation~\eqref{#1}}

\providecommand{\propsaref}[2]{Propositions~\ref{#1} and \ref{#2}}

\providecommand{\defnsaref}[2]{Definitions~\ref{#1} and \ref{#2}}
\providecommand{\colsaref}[2]{Corollaries~\ref{#1} and \ref{#2}}
\providecommand{\colssaref}[3]{Corollaries~\ref{#1}, \ref{#2} and \ref{#3}}

\DeclareSymbolFont{extraitalic}      {U}{zavm}{m}{it}
\DeclareMathSymbol{\Qoppa}{\mathord}{extraitalic}{161}
\DeclareMathSymbol{\qoppa}{\mathord}{extraitalic}{162}
\DeclareMathSymbol{\Stigma}{\mathord}{extraitalic}{167}
\DeclareMathSymbol{\Sampi}{\mathord}{extraitalic}{165}
\DeclareMathSymbol{\sampi}{\mathord}{extraitalic}{166}
\DeclareMathSymbol{\stigma}{\mathord}{extraitalic}{168}

\title{\textsc{The BRST Double Complex for the Coupling\\of Gravity to Gauge Theories}}
\author{David Prinz\footnote{Department of Mathematics and Department of Physics at Humboldt University of Berlin and Department of Mathematics at University of Potsdam; prinz@\{math.hu-berlin.de, physik.hu-berlin.de, math.uni-potsdam.de\}}}
\date{June 1, 2025}

\begin{document}

\maketitle

\begin{abstract}
	We consider (effective) Quantum General Relativity coupled to the Standard Model (QGR-SM) and clarify whether graviton-ghosts couple to matter particles. To this end, we examine the corresponding BRST and anti-BRST symmetries, which are generated by infinitesimal diffeomorphisms and infinitesimal gauge transformations. In particular, we study their properties and relations: We find that all differentials mutually anticommute, which implies that they form a double complex. In particular, we introduce the \emph{total BRST differential} as the sum of the diffeomorphism and gauge BRST differentials and similarly the \emph{total anti-BRST differential} as the sum of the respective anti-BRST differentials. Furthermore, we identify the functionals in particle fields that are (co)cycles up to total derivatives with respect to the diffeomorphism differentials as scalar tensor densities of weight one: This implies that graviton-ghosts decouple from matter particles if and only if the Yang--Mills gauge fixing Lagrange density has said tensor density weight. Moreover, we discuss the relevant gauge fixing fermions: Starting from the de Donder and Lorenz gauge fixing conditions, we introduce a \emph{total gauge fixing fermion} that generates the complete gauge fixing and ghost Lagrange density of QGR-SM. Finally, we show that the BRST cocomplexes are isomorphic to their corresponding anti-BRST complexes via ghost conjugation. Notably, this relates the BRST cohomologies to their respective anti-BRST homologies.
\end{abstract}

\section{Introduction} \label{sec:introduction}

BRST cohomology is a powerful tool to study quantum gauge theories together with their gauge fixings and corresponding ghosts via homological algebra \cite{Becchi_Rouet_Stora_1,Becchi_Rouet_Stora_2,Tyutin,Becchi_Rouet_Stora_3}. More precisely, a nilpotent operator \(S\) is introduced that performs an infinitesimal gauge transformation in direction of the ghost field. This so-called BRST operator \(S\) can be seen either as an odd vector field on the super vector bundle of particle fields or as an odd derivation on the superalgebra of particle fields. The nilpotency of \(S\) can then be used to compute its cohomology. This is useful, as physical states of the system can be identified with elements in its zeroth cohomology class. Furthermore, this formalism can be used to unify the gauge fixing and ghost Lagrange densities as follows: First of all, we understand a quantum gauge theory Lagrange density \(\mathcal{L}_\text{QGT}\) as the sum of the classical gauge theory Lagrange density \(\mathcal{L}_\text{GT}\) together with a gauge fixing Lagrange density \(\mathcal{L}_\text{GF}\) and its accompanying ghost Lagrange density \(\mathcal{L}_\text{Ghost}\), i.e.\
\begin{equation}
	\mathcal{L}_\text{QGT} \coloneq \mathcal{L}_\text{GT} + \mathcal{L}_\text{GF} + \mathcal{L}_\text{Ghost} \, .
\end{equation}
This terminology is motivated by the fact that perturbative quantization requires a gauge fixing term to calculate the propagator and a ghost term to ensure transversality. By construction, the gauge fixing and ghost Lagrange densities are not independent: In the Faddeev--Popov setup the ghost Lagrange density is designed such that the ghost field satisfies residual gauge transformations of the chosen gauge fixing as equations of motion, with the antighost acting as Lagrange multiplier.\footnote{We remark that it is possible to extend this setup by anti-BRST operators, so that ghosts and antighosts can be treated on an equal footing, cf.\ \cite{Baulieu_Thierry-Mieg_1,Prinz_6} and the constructions and results of the present article.} In the BRST framework, both the gauge fixing and ghost Lagrange densities can be generated from a so-called gauge fixing fermion \(\chi\) in ghost-degree -1 via the action of \(S\), i.e.\
\begin{equation}
	\mathcal{L}_\text{GF} + \mathcal{L}_\text{Ghost} \equiv S \chi \, .
\end{equation}
Since the term \(S \chi\) is \(S\)-exact, it is also \(S\)-closed and thus does not contribute to the zeroth cohomology class. Thus, in particular, it does not affect physical observables. To incorporate the gauge fixing, we additionally add the corresponding Lautrup--Nakanishi auxiliary fields \cite{Nakanishi,Lautrup}: These are Lie algebra valued fields that act as Lagrange multipliers and whose equations of motion are given by the gauge fixing conditions. Furthermore, it is then possible to complement this framework via an anti-BRST operator \(\overline{S}\), which is a homological differential, by essentially replacing ghosts with antighosts in addition to a slightly modified action on the corresponding ghost, antighost and Lautrup--Nakanishi auxiliary field, cf.\ \cite{Baulieu_Thierry-Mieg_1,Prinz_6,Nakanishi_Ojima,Faizal} and the present article. This setup then allows us to generate a special class of gauge fixing fermions via the action of \(\overline{S}\) on a so-called gauge fixing boson \(W\) in ghost-degree zero, i.e.\
\begin{equation}
	\omega \equiv \overline{S} W \, , \label{eqn:gauge-fixing-boson}
\end{equation}
which produces symmetric (i.e.\ Hermitian) ghost Lagrange densities such that the antighost is actually the antiparticle of the ghost, cf.\ \cite{Baulieu_Thierry-Mieg_1,Prinz_6}.

We refer the interested reader to the general introductory texts on BRST symmetry \cite{Barnich_Brandt_Henneaux_BRST,Mnev,Wernli} and the historical overview \cite{Becchi}. In addition, we mention previous studies of perturbative quantum gravity \cite{Baulieu_Thierry-Mieg_2,Baulieu_Bellon_1,Baulieu_Bellon_2,Moritsch_Schweda_Sorella,Barnich_Brandt_Henneaux_EYM,Piguet,Upadhyay,Shestakova} and the generalization of this setup to include anti-BRST operators \cite{Baulieu_Thierry-Mieg_1,Prinz_6,Nakanishi_Ojima,Faizal}. Furthermore, we emphasize that this article deals with the situation in perturbative Quantum Field Theory --- the corresponding situation in Algebraic Quantum Field Theory is discussed in \cite{Brunetti_Fredenhagen_Rejzner}.\footnote{In particular, since our observables are cross sections --- i.e.\ probabilities for \(n\)-point functions, rather than diffeomorphism-invariant functionals --- the main result of \cite{Brunetti_Fredenhagen_Rejzner} does not affect our studies: In fact, the appropriate notion for our setup is the transversality of said \(n\)-point amplitudes, cf.\ \cite{Prinz_7} and the references therein.}

In recent articles, we have studied several aspects of (effective) Quantum General Relativity coupled to the Standard Model: This includes a proper treatment of its geometric foundations \cite{Prinz_2}, a generalization of the Connes--Kreimer renormalization framework to gauge theories and gravity \cite{Prinz_3} and the complete gravity-matter Feynman rules \cite{Prinz_4}. In this article, we introduce and study the corresponding BRST double complex: The invariance of the theory under diffeomorphisms and gauge transformations implies first of all the existence of two such operators, \(P\) and \(Q\): The first performs infinitesimal diffeomorphisms in direction of the graviton-ghost and the second performs infinitesimal gauge transformations in direction of the gauge ghost, cf.\ \defnsaref{defn:diffeomorphism_brst_operator}{defn:gauge_brst_operator}. Then we provide the two gauge fixing fermions \(\stigma\) and \(\digamma\): The first implements the de Donder gauge fixing condition together with its corresponding graviton-ghosts and the second implements the Lorenz gauge fixing condition together with its corresponding gauge ghosts, cf.\ \defnsaref{prop:de_donder_gauge_fixing_fermion}{prop:lorenz_gauge_fixing_fermion}. In particular, we have reworked the conventions such that the quadratic gauge fixing and ghost Lagrange densities are both rescaled by the inverses of the respective gauge fixing parameters, \(\zeta\) and \(\xi\): This implies that all unphysical propagators are rescaled by these parameters, which introduces an additional grading on the algebra of Feynman diagrams, cf.\ \cite{Prinz_7,Prinz_3,Prinz_9}. Furthermore, we show that all non-constant functionals in the superalgebra of particle fields that are essentially closed with respect to the diffeomorphism BRST operator \(P\) and the diffeomorphism anti-BRST operator \(\overline{P}\) are scalar tensor densities of weight \(w = 1\), cf.\ \lemref{lem:p_tensor_densities}: This allows us to show that graviton-ghosts decouple from matter of the Standard Model if the gauge fixing fermion of Yang--Mills theory is a tensor density of weight \(w = 1\), cf.\ \thmref{thm:no-couplings-grav-ghost-matter-sm}. In particular, we show that every such gauge fixing fermion can be modified uniquely to satisfy said condition. Moreover, we introduce the corresponding anti-BRST operators \(\overline{P}\) and \(\overline{Q}\) in \defnsaref{defn:diffeomorphism_anti-brst_operator}{defn:gauge_anti-brst_operator} and show that all BRST operators mutually anticommute, i.e.\
\begin{subequations}
\begin{align}
	\commutatorbig{P}{\overline{P}} = \commutatorbig{\overline{P}}{\overline{P}} = \commutatorbig{Q}{\overline{Q}} = \commutatorbig{\overline{Q}}{\overline{Q}} & = 0
	\intertext{and}
	\commutatorbig{P}{Q} = \commutatorbig{P}{\overline{Q}} = \commutatorbig{\overline{P}}{Q} = \commutatorbig{\overline{P}}{\overline{Q}} & = 0 \, ,
\end{align}
\end{subequations}
cf.\ \propsaref{prop:p-cohomological-vector-field}{prop:q-cohomological-vector-field}, \colsaref{col:anti-diffeomorphism_brst_operator}{col:anti-gauge_brst_operator}, \thmref{thm:total_brst_operator} and \colref{col:total_anti-brst_operator}.\footnote{We emphasize that we use the symbol \(\left [ \, \cdot \, , \cdot \, \right ]\) for the supercommutator: In particular, it denotes the anticommutator if both arguments are odd, cf.\ \defnref{defn:supercommutator}.} This is a non-trivial observation, as infinitesimal diffeomorphisms concern all particle fields and thus in particular the operators \(Q\) and \(\overline{Q}\). As a result, their sums
\begin{subequations}
\begin{align}
	D & \coloneq P + Q
	\intertext{and}
	\overline{D} & \coloneq \overline{P} + \overline{Q}
\end{align}
\end{subequations}
are also differentials, which we call \emph{total BRST operator} and \emph{total anti-BRST operator}. This allows us to identify the physical states of the theory as elements in the respective zeroth (co)homology classes. Furthermore, we show that the sum of the gauge fixing fermion \(\stigma^{(1)}\) for the linearized de Donder gauge fixing and the gauge fixing fermion \(\digamma \! \! _{\{ 1 \}}\) for the covariant Lorenz gauge fixing,
\begin{equation}
	\Upsilon \coloneq \stigma^{(1)} + \digamma \! \! _{\{ 1 \}} \, ,
\end{equation}
is again a gauge fixing fermion, which we call \emph{total gauge fixing fermion}, cf.\ \thmref{thm:total_gauge_fixing_fermion}.\footnote{We remark that this works independently of the chosen gauge fixing conditions, as long as the gauge fixing fermion of the gauge theory is a tensor density of weight \(w = 1\).} In particular, we obtain the complete gauge fixing and ghost Lagrange density of (effective) Quantum General Relativity coupled to the Standard Model via \(D \Upsilon\). We believe that this analysis provides an important contribution to the quantization of gravity coupled to gauge theories. Additionally, we relate the BRST cocomplexes to the anti-BRST complexes via the ghost-conjugation involutions in \thmref{thm:relation_cocomplex_complex}:
\begin{subequations}
\begin{align}
	\bigl ( \CQ^{i,j}, P^{i,j} \bigr )^{\dagger_C} & \cong \bigl ( \CQ_{-i}^j, \overline{P}_{-i}^j \bigr ) \, , \\
	\bigl ( \CQ^{i,j}, Q^{i,j} \bigr )^{\dagger_c} & \cong \bigl ( \CQ^i_{-j}, \overline{Q}^i_{-j} \bigr )
	\intertext{and}
	\bigl ( \CQ^k, D^k \bigr )^\dagger & \cong \bigl ( \CQ_{-k}, \overline{D}_{-k} \bigr ) \, ,
\end{align}
\end{subequations}
for all \(i, j, k \in \mathbb{Z}\), where \(i\) denotes the graviton-ghost degree, \(j\) the gauge ghost degree and \(k\) the total ghost degree. In particular, this relates the cohomology of the BRST operators to the homology of the anti-BRST operators, as we notice in \colref{col:relation_cohomology_homology}:
\begin{subequations}
\begin{align}
	H^{i,j} \bigl ( P \bigr ) & \cong H_{-i}^j \bigl ( \overline{P} \bigr ) \, , \\
	H^{i,j} \bigl ( Q \bigr ) & \cong H^i_{-j} \bigl ( \overline{Q} \bigr )
	\intertext{and}
	H^k \bigl ( D \bigr ) & \cong H_{-k} \bigl ( \overline{D} \bigr ) \, ,
\end{align}
\end{subequations}
for all \(i, j, k \in \mathbb{Z}\), where again \(i\) denotes the graviton-ghost degree, \(j\) the gauge ghost degree and \(k\) the total ghost degree. Finally, we consider the sign-twisted anti-BRST operators
\begin{equation}
	\widetilde{S}^l \coloneq \left ( -1 \right )^l \overline{S}^l \, ,
\end{equation}
where \(\overline{S} \in \setbig{\overline{P}, \overline{Q}, \overline{D}}\) denotes any of the three anti-BRST operators and \(l \in \mathbb{Z}\) the respective degree. Then, we show the following in \propref{prop:cochain_chain_homotopy}: The sign-twisted anti-BRST operators \(\widetilde{S}\) are cochain homotopies for the BRST cocomplex between the maps \(\widetilde{\mathcal{S}} \coloneq S \circ \widetilde{S}\) and the zero map, where \(S \in \setbig{P, Q, D}\) is the corresponding BRST operator. Likewise, the BRST operators are a chain homotopy for the sign-twisted anti-BRST complex between said maps. This enhances our understanding of gauge fixing fermions coming from gauge fixing bosons in the sense of \eqnref{eqn:gauge-fixing-boson}, as will be discussed in \cite{Prinz_6}.

In this article, we consider (effective) Quantum General Relativity coupled to the Standard Model, given via the following Lagrange density:
\begin{equation}
	\mathcal{L}_\text{QGR-SM} \coloneq \mathcal{L}_\text{QGR} + \mathcal{L}_\text{QYM} + \mathcal{L}_\text{Matter}
\end{equation}
Specifically, the Lagrange density for (effective) Quantum General Relativity is given as follows:
\begin{equation} \label{eqn:qgr_lagrange_density_introduction}
\begin{split}
	\mathcal{L}_\text{QGR} & \coloneq \mathcal{L}_\text{GR} + \mathcal{L}_\text{GR-GF} + \mathcal{L}_\text{GR-Ghost} \\ & \phantom{:} = - \frac{1}{2 \varkappa^2} \left ( \sqrt{- \dt{g}} R + \frac{1}{2 \zeta} \eta^{\mu \nu} \deDonder^{(1)}_\mu \deDonder^{(1)}_\nu \right ) \dif V_\eta \\ & \phantom{\coloneq} - \frac{1}{2} \eta^{\mu \nu} \overline{C}^\rho \left ( \frac{1}{\zeta} \left ( \partial_\mu \partial_\nu C_\rho \right ) + \partial_\rho \bigl ( \Gamma_{\sigma \mu \nu} C^\sigma \bigr ) - 2 \partial_\mu \bigl ( \Gamma_{\sigma \nu \rho} C^\sigma \bigr ) \right ) \dif V_\eta
\end{split}
\end{equation}
In particular, we consider the metric expansion \(g_{\mu \nu} \equiv \eta_{\mu \nu} + \varkappa h_{\mu \nu}\), where \(h_{\mu \nu}\) is the graviton field and \(\varkappa \coloneq \sqrt{\kappa}\) the graviton coupling constant (with \(\kappa \coloneq 8 \pi G\) the Einstein gravitational constant). In addition, \(R \coloneq g^{\nu \sigma} \tensor{R}{^\mu _\nu _\mu _\sigma}\) is the Ricci scalar with \(\tensor{R}{^\rho _\sigma _\mu _\nu} \coloneq \partial_\mu \tensor{\Gamma}{^\rho _\nu _\sigma} - \partial_\nu \tensor{\Gamma}{^\rho _\mu _\sigma} + \tensor{\Gamma}{^\rho _\mu _\lambda} \tensor{\Gamma}{^\lambda _\nu _\sigma} - \tensor{\Gamma}{^\rho _\nu _\lambda} \tensor{\Gamma}{^\lambda _\mu _\sigma}\) the Riemann tensor and \(\tensor{\Gamma}{^\rho _\mu _\nu} \coloneq g^{\rho \sigma} \textfrac{\left ( \partial_\mu g_{\sigma \nu} + \partial_\nu g_{\mu \sigma} - \partial_\sigma g_{\mu \nu} \right )}{2}\) the Christoffel symbol for the Levi-Civita connection. Furthermore, \(\dif V_\eta \coloneq \dif t \wedge \dif x \wedge \dif y \wedge \dif z\) denotes the Minkowskian volume form, which is related to the Riemannian volume form via \(\dif V_g\) via \(\dif V_g \equiv \sqrt{- \dt{g}} \dif V_\eta\). Moreover, \(\deDonder^{(1)}_\mu \coloneq \eta^{\rho \sigma} \Gamma_{\mu \rho \sigma} \equiv 0\) is the linearized de Donder gauge fixing functional with \(\Gamma_{\mu \rho \sigma} \coloneq \textfrac{\varkappa \left ( \partial_\rho h_{\mu \sigma} + \partial_\sigma h_{\rho \mu} - \partial_\mu h_{\rho \sigma} \right )}{2}\) and \(\zeta\) denotes the gauge fixing parameter. Finally, \(C_\rho\) and \(\overline{C}^\rho\) are the graviton-ghost and graviton-antighost, respectively. We refer to \cite{Prinz_2,Prinz_4} for more detailed introductions and further comments on the chosen conventions. Then, additionally, the Lagrange density for Quantum Yang--Mills theory is given as follows:
\begin{equation} \label{eqn:qym_lagrange_density_introduction}
\begin{split}
	\mathcal{L}_\text{QYM} & \coloneq \mathcal{L}_\text{YM} + \mathcal{L}_\text{YM-GF} + \mathcal{L}_\text{YM-Ghost} \\ & \phantom{:} = - \frac{1}{2 \mathrm{g}^2} \delta_{a b} \left ( \frac{1}{2} g^{\mu \nu} g^{\rho \sigma} F^a_{\mu \rho} F^b_{\nu \sigma} + \frac{1}{\xi} \Lorenz^a \Lorenz^b \right ) \dif V_g \\
	& \phantom{\coloneq} - g^{\mu \nu} \left ( \frac{1}{\xi} \overline{c}_a \bigl ( \nabla^{TM}_\mu \partial_\nu c^a \bigr ) + \mathrm{g} \tensor{f}{^a _b _c} \overline{c}_a \left ( \nabla^{TM}_\mu \bigl ( c^b A^c_\nu \bigr ) \right ) \right ) \dif V_g
\end{split}
\end{equation}
We remark that \(F^a_{\mu \nu} \coloneq \mathrm{g} \bigl ( \partial_\mu A^a_\nu - \partial_\nu A^a_\mu \bigr ) - \mathrm{g}^2 \tensor{f}{^a _b _c} A^b_\mu A^c_\nu\) is the local curvature form of the gauge boson \(A^a_\mu\). Additionally, \(\Lorenz^a \coloneq \mathrm{g} g^{\mu \nu} \bigl ( \nabla^{TM}_\mu A^a_\nu \bigr ) \equiv 0\) is the covariant Lorenz gauge fixing functional and \(\xi\) denotes the gauge fixing parameter. Finally, \(c^a\) and \(\overline{c}_a\) are the gauge ghost and gauge antighost, respectively. These two Lagrange densities are then completed with the matter Lagrange densities for a vector of complex scalar fields and a vector of spinor fields, both subjected to the action of the gauge group. Explicitly, the matter Lagrange density is given as follows:
\begin{equation} \label{eqn:matter_lagrange_density_introduction}
	\begin{split}
	\mathcal{L}_\text{Matter} & \coloneq \mathcal{L}_\text{Higgs} + \mathcal{L}_\text{Fermion} \\
	& \phantom{:} = \left ( g^{\mu \nu} \left ( \nabla^{H}_\mu \Phi \right )^\dagger \left ( \nabla^{H}_\nu \Phi \right ) + \sum_{i \in \boldsymbol{I}_\Phi} \frac{\alpha_i}{i!} \bigl ( \Phi^\dagger \Phi \bigr )^i + \overline{\Psi} \left ( \imaginary \slashed{\nabla}^{\boldsymbol{\Sigma} M} - \boldsymbol{m}_\Psi \right ) \Psi \right ) \dif V_g
	\end{split}
\end{equation}
Here, \(\Phi\) and \(\Psi\) denote the respective vectors of complex scalar fields and spinor fields, with corresponding dual vectors \(\Phi^\dagger\) and dual spinors \(\overline{\Psi} \coloneq \left ( \boldsymbol{\gamma}_0 \Psi \right )^\dagger\). Furthermore, \(\nabla^{H}_\mu \coloneq \partial_\mu + \imaginary \mathrm{g} A^a_\mu \mathfrak{H}_a\) and \(\nabla^{\boldsymbol{\Sigma} M}_\mu \coloneq \partial_\mu + \boldsymbol{\varpi}_\mu + \imaginary \mathrm{g} A^a_\mu \mathfrak{S}_a\) denote the respective covariant derivatives, where \(\mathfrak{H}_a\) and \(\mathfrak{S}_a\) denote the infinitesimal actions of the gauge group \(G\) on the Higgs bundle \(H\) and the twisted spinor bundle \(\boldsymbol{\Sigma} M\), respectively, and \(\boldsymbol{\varpi}_\mu\) is the spin connection on the twisted spinor bundle. In addition, \(\slashed{\nabla}^{\boldsymbol{\Sigma} M} \coloneq e^{\mu m} \gamma_m \bigl ( \partial_\mu + \boldsymbol{\varpi}_\mu + \imaginary \mathrm{g} A^a_\mu \mathfrak{S}_a \bigr )\) denotes the corresponding twisted Dirac operator, where \(e^{\mu m}\) is the inverse vielbein and \(\gamma_m\) the Minkowski space Dirac matrix. Moreover, \(\boldsymbol{I}_\Phi\) denotes the set of scalar field interactions, with respective coupling constants \(\alpha_i\) (and possible mass \(\alpha_2 \coloneq - m_\Phi\)). Finally, \(\boldsymbol{m}_\Psi\) denotes the diagonal matrix with all fermion masses as entries. We refer to \cite[Subsection 4.2]{Prinz_4} for a detailed discussion thereon.

We will continue this topic in future work as follows: First we will use the BRST and anti-BRST operators given in this article to derive symmetric Lagrange densities for General Relativity and covariant Yang--Mills theory in \cite{Prinz_6}. Then, we study the transversality of the corresponding quantized theories in \cite{Prinz_7}. In addition, we also considered the case of perturbative Quantum Gravity with a cosmological constant \cite{Prinz_8}. Finally, we will also investigate on the corresponding cancellation identities via the introduction of \emph{perturbative BRST cohomology} in \cite{Prinz_9}: This will be a modified version of the Feynman graph cohomology introduced by Kreimer et al.\ in the realm of the Corolla polynomial \cite{Kreimer_Yeats,Kreimer_Sars_vSuijlekom,Sars_PhD,Prinz_1,Kreimer_Corolla,Berghoff_Knispel}.

This article originates from the author's dissertation \cite{Prinz_PhD}.

\enter

\section{Geometric setup and particle fields} \label{sec:geometric-setup}

We start this article with a section on the geometric underpinnings of the BRST symmetry, specifically graded supergeometry. In this language, the BRST operator can be seen as a cohomological super vector field on the spacetime-matter bundle. Equivalently, it can also be seen as a cohomological superderivation on the algebra of particle fields. Then, we discuss spacetimes and the spacetime-matter bundle as a vector bundle whose sections describe the particle fields of (effective) Quantum General Relativity coupled to the Standard Model. After that, we define metric decompositions and the graviton field. Additionally, we provide a discussion on the diffeomorphism and gauge groups together with their infinitesimal actions. Then, we discuss the geometric background for the relation between BRST and anti-BRST symmetries. To this end, we start with the ghost conjugation as a Hermitian involution on the space of graded superfunctions. Finally, we explain the geometric reason why the BRST operators anticommute with the anti-BRST operators, as will be proven in the following sections.

\enter

\begin{defn}[\(\mathbb{Z}^2\)-graded supermanifold] \label{defn:ztwo-graded-supermanifold}
	Let \(\mathcal{M}\) be a topological manifold. We call \(\mathcal{M}\) a \(\mathbb{Z}^2\)-graded supermanifold, if it is isomorphic to a vector bundle \(\pi \colon \mathcal{M} \to M\) that splits into a direct sum bundle such that the following diagram commutes
	\begin{equation}
	\begin{tikzcd}[row sep=huge]
		\mathcal{M} \arrow[swap]{dr}{\pi} \arrow{rr}{\cong} & & \displaystyle \bigoplus_{(i,j) \in \mathbb{Z}^2} \mathcal{M}_{(i,j)} \arrow{dl}{\tilde{\pi}} \\
		& M &
	\end{tikzcd} \, ,
	\end{equation}
	where \((i,j) \in \mathbb{Z}^2\) denotes the degree of the subbundles and \(\mathcal{M}_{(0,0)} \cong M\), i.e.\ the degree \((0,0)\) is concentrated in the so-called body \(M\). We call the first integer \(i\) the graviton-ghost degree, the second integer \(j\) the gauge ghost degree and their sum \(k \coloneq i + j\) the total ghost degree. Additionally, we call the grading \emph{compatible with the super structure} of \(\mathcal{M}\) if the parity of every subbundle is given via
	\begin{equation}
		p \equiv i + j \quad \text{Mod 2} \, , \label{eqn:compatibility_grading_super_structure}
	\end{equation}
	where \(0 \in \mathbb{Z}_2\) denotes even coordinates and \(1 \in \mathbb{Z}_2\) denotes odd coordinates. Concretely, on the level of graded super functions \(\mathcal{C} \left ( \mathcal{U} \right )\) for \(\mathcal{U} \subseteq \mathcal{M}\) this means that
	\begin{equation}
		\mathcal{C} \bigl ( \mathcal{U}_{(i,j)} \bigr ) \cong \begin{cases} C^\infty \bigl ( \mathcal{U}_{(0,0)} \bigr ) & \text{if \((i,j) = (0,0)\)} \\ \mathcal{S} \bigl ( \mathcal{U}_{(i,j)} \bigr ) & \text{if \(p = 0\)} \\ \mathcal{A} \bigl ( \mathcal{U}_{(i,j)} \bigr ) & \text{if \(p = 1\)} \end{cases} \, ,
	\end{equation}
	where \(\mathcal{U}_{(i,j)} \subseteq \mathcal{M}_{(i,j)}\) is an open subset, \(C^\infty \bigl ( \mathcal{U}_{(0,0)} \bigr )\) denotes smooth functions on \(U \subseteq M\), \(\mathcal{S} \bigl ( \mathcal{U}_{(i,j)} \bigr )\) denotes symmetric formal power series (i.e.\ a formal power series in commuting variables) and \(\mathcal{A} \bigl ( \mathcal{U}_{(i,j)} \bigr )\) denotes antisymmetric formal power series (i.e.\ a formal power series in anticommuting variables). Finally, we define the grade shift via
	\begin{equation}
		\mathcal{M}_{(i,j)} [m,n] \coloneq \mathcal{M}_{(i+m,j+n)} \, ,
	\end{equation}
	which additionally implies a potential shift in parity according to \eqnref{eqn:compatibility_grading_super_structure}. We refer to \cite{Voronov} for more details in this direction.
\end{defn}

\enter

\begin{defn}[Supercommutator] \label{defn:supercommutator}
	Let \(\mathcal{M}\) be a supermanifold and \(\boldsymbol{X}_1, \boldsymbol{X}_2 \in \mathfrak{X} \left ( \mathcal{M} \right )\) be two super vector fields of distinct parities \(p_1, p_2 \in \mathbb{Z}_2\). Then we introduce the supercommutator as follows:
	\begin{equation}
		\commutatorbig{\boldsymbol{X}_1}{\boldsymbol{X}_2} \coloneq \boldsymbol{X}_1 \left ( \boldsymbol{X}_2 \right ) - \left ( -1 \right )^{p_1 p_2} \boldsymbol{X}_2 \left ( \boldsymbol{X}_1 \right )
	\end{equation}
	This turns the module \((\mathfrak{X} \left ( \mathcal{M} \right ) \! , \left [ \, \cdot \, , \cdot \, \right ])\) into a Lie superalgebra.
\end{defn}

\enter

\begin{ass}
	In the following, we assume that the grading is compatible with the super structure in the sense of \eqnref{eqn:compatibility_grading_super_structure}, such that the parity is implied by the grading.
\end{ass}

\enter

\begin{defn}[Homological and cohomological vector fields]
	Let \(\mathcal{M}\) be a \(\mathbb{Z}\)-graded supermanifold with compatible super structure. We denote the subspace of pure super vector fields \(\boldsymbol{X}^\mu\) with degree \(z \in \mathbb{Z}\) by \(\mathfrak{X}_{(z)} \left ( \mathcal{M} \right )\). Then an odd vector field \(\Xi \in \mathfrak{X} \left ( \mathcal{M} \right )\) with the property
	\begin{equation}
	\commutatorbig{\Xi}{\Xi} \equiv 2 \, \Xi^2 \equiv 0
	\end{equation}
	is called homological if \(\Xi \in \mathfrak{X}_{(-1)} \left ( \mathcal{M} \right )\) and cohomological if \(\Xi \in \mathfrak{X}_{(1)} \left ( \mathcal{M} \right )\). This turns \((\mathcal{C}_\bullet \left ( \mathcal{M} \right ),\Xi)\) into a (co)chain complex and the pair \((\mathcal{M}, \Xi)\) is called a differential-graded manifold.
\end{defn}

\enter

\begin{exmp}
	Let \(M\) be a manifold with \(\Omega^\bullet \left ( M \right ) \coloneq \Gamma \left ( M, \bigwedge^\bullet TM \right )\) its sheaf of differential forms. Let furthermore \(\mathcal{M} \coloneq T[1]M\) denote its degree shifted tangent bundle. Then we can identify \(\mathcal{C}^\bullet \left ( \mathcal{M} \right ) \cong \Omega^\bullet \left ( M \right )\), where the grading is now given by the form degree. With this, we obtain a cohomological vector field via the de Rham differential \(\mathrm{d} \in \mathfrak{X}_{(1)} \left ( \mathcal{M} \right )\) and a homological vector field via the de Rham codifferential \(\delta \in \mathfrak{X}_{(-1)} \left ( \mathcal{M} \right )\).
\end{exmp}

\enter

\begin{defn}[Spacetime] \label{def:spacetime}
	Let \((M,g)\) be a \(d\)-dimensional Lorentzian manifold. We call \((M,g)\) a spacetime, if it is smooth, connected and time-orientable.
\end{defn}

\enter

\begin{defn}[Spacetime-matter bundle] \label{defn:spacetime-matter_bundle}
	Let \((M, g)\) be a \(d\)-dimensional spacetime. Then we define the spacetime-matter bundle of (effective) Quantum General Relativity coupled to the Standard Model as the globally trivial \(\mathbb{Z}^2\)-graded super bundle \(\beta_\Q \colon \BQ \to M\), where \(\BQ \coloneq M \times_M \mathcal{V}_\Q\) is the fiber product over \(M\) with
	\begin{equation}
	\begin{split}
		\mathcal{V}_\Q & \coloneq \left ( \operatorname{Sym}^2_\mathbb{R} \left ( T^* M \right ) \right )^{\times 3} \times \left ( T^* M \otimes_\mathbb{R} E \right ) \times \Bigl ( T^* [1,0] M \oplus T [-1,0] M \oplus T M \Bigr ) \\ & \phantom{\coloneq} \times \left ( T^*M \otimes_\mathbb{R} \mathfrak{g} \right ) \times \left ( G \times_\rho \left ( H^{(i)} \oplus \Sigma M^{\oplus j} \right ) \right ) \times \Bigl ( \mathfrak{g} [0,1] \oplus \mathfrak{g}^* [0,-1] \oplus \mathfrak{g}^* \Bigr ) \, , \label{eqn:spacetime-matter_bundle}
	\end{split}
	\end{equation}
	where we have the following affine bundles:
	\begin{itemize}
		\item Metric, background metric and graviton field as a section in the triple Cartesian product \(\bigl ( \operatorname{Sym}^2_\mathbb{R} \left ( T^* M \right ) \bigr )^{\times 3} \coloneq \bigtimes_{m = 1}^3 \bigl ( \operatorname{Sym}^2_\mathbb{R} \left ( T^* M \right ) \bigr )\), where \(\operatorname{Sym}^2_\mathbb{R} \left ( T^* M \right ) \coloneq \left ( T^* M \otimes_\mathbb{R} T^* M \right ) / \mathbb{Z}_2\) is the symmetrized tensor product
		\item Vielbein field as a section in \(T^* M \otimes_\mathbb{R} E\), where \(E\) is a real \(d\)-dimensional vector bundle
		\item Graviton-ghost as a section in \(T^* [1,0] M\)
		\item Graviton-antighost as a section in \(T [-1,0] M\)
		\item Graviton-Lautrup--Nakanishi field as a section in \(T M\)
		\item Gauge bosons as a section in \(T^*M \otimes_\mathbb{R} \mathfrak{g}\)
		\item Higgs and Goldstone bosons as sections in the fiber \(G \times_{\rho_H} H^{(i)}\), where \(\rho_H \colon G \to H^{(i)}\) is the action of the gauge group \(G\) on the Higgs bundle \(H^{(i)} \coloneq \mathbb{C}^i\)
		\item Fermion families as sections in \(G \times_{\rho_{\Sigma M}} \Sigma M^{\oplus j}\), where \(\rho_{\Sigma M} \colon G \to \Sigma M^{\oplus j}\) is the action of the gauge group \(G\) on the Whitney sum of \(j\) spinor bundles \(\Sigma M^{\oplus j} \coloneq \bigoplus_{n = 1}^j \Sigma M\)
		\item Gauge ghost as a section in the bundle with fiber \(\mathfrak{g} [0,1]\)
		\item Gauge antighost as a section in the bundle with fiber \(\mathfrak{g}^* [0,-1]\)
		\item Gauge Lautrup--Nakanishi field as a section in the bundle with fiber \(\mathfrak{g}^*\)
	\end{itemize}
	Here, the ghosts are odd sections of either graviton-ghost degree \(\pm 1\) or gauge ghost degree \(\pm 1\), respectively.
\end{defn}

\enter

\begin{defn}[Sheaf of particle fields] \label{defn:sheaf-of-particle-fields}
	Let \((M, g)\) be a spacetime with topology \(\mathcal{T}_M\) and \(\beta_\Q \colon \BQ \to M\) the spacetime-matter bundle from \defnref{defn:spacetime-matter_bundle}. Then we define the sheaf of particle fields via
	\begin{equation}
		\FQ \, : \quad \mathcal{T}_M \to \Gamma \left ( M, \BQ \right ) \, , \quad U \mapsto \Gamma \left ( U, B \right ) \, ,
	\end{equation}
	where \(B \subset \BQ\) is one of the subbundles from \eqnref{eqn:spacetime-matter_bundle}. More precisely, we consider the following particle fields:
	\begin{itemize}
		\item Lorentzian metrics \(g \in \operatorname{LorMet} \left ( M \right ) \subset \Gamma \bigl ( M, \operatorname{Sym}^2_\mathbb{R} \left ( T^* M \right ) \bigr )\)
		\item Minkowski background metric \(\eta \in \operatorname{LorMet} \left ( M \right ) \subset \Gamma \bigl ( M, \operatorname{Sym}^2_\mathbb{R} \left ( T^* M \right ) \bigr )\)
		\item Graviton field \(\varkappa h \coloneq \left ( g - b \right ) \in \operatorname{Grav} \left ( M \right ) \subset \Gamma \bigl ( M, \operatorname{Sym}^2_\mathbb{R} \left ( T^* M \right ) \bigr )\), where \(\varkappa\) is the graviton coupling constant
		\item Vielbein fields as bundle isomorphisms \(e \in \operatorname{BdlIso}_M \left ( T M, E \right ) \subset \sctnbig{M, T^* M \otimes_\mathbb{R} E}\)
		\item Vector of \(2i\) Higgs and Goldstone fields \(\Phi \in \sctnbig{M, M \times H^{(i)}}\)
		\item Vector of \(j\) fermion fields \(\Psi \in \sctnbig{M, \Sigma M^{\oplus j}}\)
		\item Gauge boson fields \(\imaginary \mathrm{g} A \in \operatorname{Conn} \left ( M, \mathfrak{g} \right ) \subset \Omega^1 \left ( M, \mathfrak{g} \right )\), where \(\imaginary \coloneq \sqrt{-1}\) is the imaginary unit and \(\mathrm{g}\) is the gauge boson coupling constant
		\item Graviton-ghost fields \(C \in \sctnbig{M, T^* [1,0] M}\)
		\item Graviton-antighost fields \(\overline{C} \in \sctnbig{M, T [-1,0] M}\)
		\item Graviton-Lautrup--Nakanishi auxiliary fields \(B \in \mathfrak{X} \left ( M \right )\)
		\item Gauge ghost fields \(c \in \sctnbig{M, M \times \mathfrak{g} [0,1]}\)
		\item Gauge antighost fields \(\overline{c} \in \sctnbig{M, M \times \mathfrak{g}^* [0,-1]}\)
		\item Gauge Lautrup--Nakanishi auxiliary fields \(b \in \sctnbig{M, M \times \mathfrak{g}^*}\)
	\end{itemize}
	Specifically, given a metric \(g_{\mu \nu}\) and the Minkowski background metric \(\eta_{\mu \nu}\), the graviton field \(h_{\mu \nu}\) is then defined as their difference, rescaled by the graviton coupling constant \(\varkappa \coloneq \sqrt{\kappa}\), with \(\kappa \coloneq 8 \pi G\) the Einstein constant and \(G\) the Newton constant:
	\begin{equation} \label{eqn:metric_decomposition}
		h_{\mu \nu} \coloneq \frac{1}{\varkappa} \left ( g_{\mu \nu} - \eta_{\mu \nu} \right ) \iff g_{\mu \nu} \equiv \eta_{\mu \nu} + \varkappa h_{\mu \nu} \, .
	\end{equation}
	Thus, the graviton field \(h_{\mu \nu}\) is given as a rescaled, symmetric \((0,2)\)-tensor field, i.e.\ a section \(\varkappa h \in \Gamma \bigl ( M, \operatorname{Sym}^2_\mathbb{R} \left ( T^* M \right ) \bigr )\).
\end{defn}

\enter

\begin{defn}[Functionals of particle fields]
	Let \(\FQ\) be the space of fields from \defnref{defn:sheaf-of-particle-fields}. Then we define the space of functionals as graded-symmetric polynomials in its corresponding dual space, i.e.\ \(\CQ \coloneq \operatorname{Sym} \bigl ( \FQ^\vee \bigr )\). Furthermore, we denote by \(\CQ^{\{w\}, (i,j)}\) its subspace in graviton-ghost degree \(i \in \mathbb{Z}\), gauge ghost degree \(j \in \mathbb{Z}\) and tensor density weight \(w \in \mathbb{R}\).\footnote{I.e.\ \(\mathfrak{f} \equiv \left ( - \dt{g} \right )^{w/2} f\) for an ordinary covariant functional \(f \in \CQ^{\{0\}, (i,j)}\) in said bidegree.} If we are only interested in the total ghost degree \(k \in \mathbb{Z}\), we also write \(\CQ^{\{w\}, (k)}\). Likewise, if the tensor density weight is irrelevant for the current statement, we simply write \(\CQ^{(i,j)}\) or \(\CQ^{(k)}\), respectively. We emphasize that we do not invert the grading when passing from sections to functionals of sections; rather we consider ghosts and antighosts to always have degree one and minus one, respectively.\footnote{In particular, this corresponds to the fact that we raise and lower indices with metrics in degree zero.}
\end{defn}

\enter

\begin{defn}[Diffeomorphism group and group of gauge transformations] \label{defn:diffeomorphism-group-and-group-of-gauge-transformations}
	Given the situation of \defnref{defn:sheaf-of-particle-fields}, the physical theories that we are studying are invariant under the action of two groups: First, the diffeomorphism group homotopic to the identity \(\mathcal{D} \coloneq \Diff\), which turns out to be compactly supported \cite{Schmeding_PhD}. Secondly, the group of gauge transformations \(\mathcal{G} \coloneq \Gamma \left ( M, M \times G \right )\), where \(G \cong U(1) \times \widetilde{G}\) is the gauge group with \(\widetilde{G}\) a compact and semisimple Lie group. The specific case of the Standard Model is given via \(\widetilde{G}_\text{SM} \cong \bigl ( SU (2) \times SU (3) \bigr ) / N\), where \(N\) a normal discrete subgroup, cf.\ \cite{Saller,Baez,Baez_Huerta,Tong}. Then, the diffeomorphism group homotopic to the identity acts via
	\begin{subequations}
	\begin{align}
		\varrho & \, : \quad \mathcal{D} \times \BQ \to \BQ \quad \left ( \phi, \varphi \right ) \mapsto \phi_* \varphi \, ,
		\intertext{where \(\varrho\) acts naturally on \(M\) and via push-forward on the corresponding particle bundles.\footnotemark \phantom{ } Furthermore, the group of gauge transformations acts fiberwise via}
		\rho & \, : \quad \mathcal{G} \times \BQ \to \BQ \quad \left ( \gamma, \varphi \right ) \mapsto \gamma \cdot \varphi \, ,
	\end{align}
	\end{subequations}
	\footnotetext{The action on the spinor bundle is more involved, as its construction depends crucially on the metric \(g\). We refer to \cite{Mueller_Nowaczyk} for an explicit construction.}%
	\noindent where \(\rho\) acts via the matrix representation on the vectors of Higgs and spinor fields. Additionally, we also consider the action of infinitesimal diffeomorphisms via
	\begin{subequations}
	\begin{align}
		\boldsymbol{\varrho} & \, : \quad \mathfrak{D} \times \BQ \to \BQ \quad \left ( X, \varphi \right ) \mapsto \Lie_X \varphi \, ,
		\intertext{where \(\mathfrak{D} \coloneq \diff \cong \mathfrak{X}_\text{c} \left ( M \right )\) is the Lie algebra of compactly supported vector fields and \(\Lie_X\) denotes the Lie derivative of the geodesic exponential map. Moreover, we also consider the action of infinitesimal gauge transformations via}
		\boldsymbol{\rho} & \, : \quad \mathfrak{G} \times \BQ \to \BQ \quad \left ( Z, \varphi \right ) \mapsto \ell_Z \varphi \, ,
	\end{align}
	\end{subequations}
	where \(\mathfrak{G} \coloneq \Gamma \left ( M, M \times \mathfrak{g} \right )\) is the Lie algebra of \(\mathfrak{g}\)-valued vector fields, with \(\mathfrak{g}\) the Lie algebra of the gauge group \(G\), and \(\ell_Z\) denotes the Lie derivative of the Lie exponential map.
\end{defn}

\enter

\begin{defn}[Transformation under (infinitesimal) diffeomorphisms] \label{defn:transformation_diffeo}
	Given the situation of \defnref{defn:diffeomorphism-group-and-group-of-gauge-transformations}, we define the action of diffeomorphisms \(\phi \in \Diff\) on the graviton field via
	\begin{subequations}
	\begin{align}
		\phi_* \left ( \varkappa h \right ) & \coloneq \phi_* g \, ,
		\intertext{such that the Minkowski background metric can be conveniently defined to be invariant, i.e.}
		\phi_* \eta & \coloneq 0 \, ,
	\end{align}
	\end{subequations}
	and on the other particle fields \(\varphi \in \sctn{M, \BQ }\) as usual, i.e.\ via \(\phi_* \varphi\). In particular, the action of infinitesimal diffeomorphisms is given via the Lie derivative with respect to its generating vector field \(X \in \diff \cong \vectc{M}\), i.e.\
	\begin{subequations}
	\begin{align}
		\boldsymbol{\varrho} \left ( X, h_{\mu \nu} \right ) & \equiv \frac{1}{\varkappa} \left ( \nabla^{TM}_\mu X_\nu + \nabla^{TM}_\nu X_\mu \right ) \, , \\
		\boldsymbol{\varrho} \left ( X, \eta_{\mu \nu} \right ) & \equiv 0
		\intertext{and}
		\boldsymbol{\varrho} \left ( X, \varphi \right ) & \equiv \Lie_X \varphi \, ,
	\end{align}
	\end{subequations}
	where \(\nabla^{TM}\) denotes the covariant derivative with respect to the Levi-Civita connection \(\Gamma\), i.e.\
	\begin{subequations}
	\begin{align}
		\tensor{\Gamma}{^\rho _\mu _\nu} & \coloneq \frac{1}{2} g^{\rho \sigma} \left ( \partial_\mu g_{\sigma \nu} + \partial_\nu g_{\mu \sigma} - \partial_\sigma g_{\mu \nu} \right )
		\intertext{and}
		\Gamma_{\sigma \mu \nu} & \coloneq \frac{1}{2} \left ( \partial_\mu g_{\sigma \nu} + \partial_\nu g_{\mu \sigma} - \partial_\sigma g_{\mu \nu} \right ) \, .
	\end{align}
	\end{subequations}
\end{defn}

\enter

\begin{rem}
	In \defnref{defn:transformation_diffeo} we have only considered diffeomorphisms homotopic to the identity, as they are generated by flows of compactly supported vector fields \(X \in \diff \cong \vectc{M}\) and thus allow for an infinitesimal picture. Notably, they differ from the identity only on compactly supported domains. Thus, diffeomorphisms homotopic to the identity preserve the asymptotic structure of spacetimes. We remark that, different from finite dimensional Lie groups, the Lie exponential map
	\begin{subequations}
	\begin{align}
		\operatorname{exp} \, & : \quad \diff \to \Diff
		\intertext{is no longer locally surjective, which leads to the notion of an evolution map}
		\operatorname{Evol} \, & : \quad C^\infty \left ( [0,1], \diff \right ) \to C^\infty \left ( [0,1], \Diff \right )
	\end{align}
	\end{subequations}
	that maps smooth curves in the Lie algebra to smooth curves in the corresponding Lie group. Specifically, the Lie exponential map is contained in the evolution map by considering constant maps. We refer to \cite{Schmeding_PhD,Prinz_Schmeding_1,Prinz_Schmeding_2} for further details.
\end{rem}

\enter

\begin{defn}[Ghost conjugation, anti-Hermitian auxiliary field] \label{defn:ghost-conjugation}
	Given the situation of \defnref{defn:sheaf-of-particle-fields}, we introduce the following three Hermitian involutions on the space of particle fields \(\FQ\), which we then multiplicatively extend to functionals: First, the graviton-ghost conjugation \(\dagger_C\) and the gauge ghost conjugation \(\dagger_c\) via (\(\varphi\) denotes again any other particle field):
	{\allowdisplaybreaks
	\begin{subequations}
	\begin{align}
		& \bigl ( C^\rho \bigr )^{\dagger_C} \coloneq \overline{C}^\rho && \bigl ( C^\rho \bigr )^{\dagger_c} \coloneq C^\rho \\
		& \bigl ( \overline{C}^\rho \bigr )^{\dagger_C} \coloneq C^\rho && \bigl ( \overline{C}^\rho \bigr )^{\dagger_c} \coloneq \overline{C}^\rho \\
		& \bigl ( B^\rho \bigr )^{\dagger_C} \coloneq - B^\rho - \varkappa \zeta \left ( \overline{C}^\sigma \bigl ( \partial_\sigma C^\rho \bigr ) - \bigl ( \partial_\sigma \overline{C}^\rho \bigr ) C^\sigma \right ) && \bigl ( B^\rho \bigr )^{\dagger_c} \coloneq B^\rho \\
		& \bigl ( {B^\prime}^\rho \bigr )^{\dagger_C} \coloneq - {B^\prime}^\rho && \bigl ( {B^\prime}^\rho \bigr )^{\dagger_c} \coloneq {B^\prime}^\rho \\
		& \bigl ( c^a \bigr )^{\dagger_C} \coloneq c^a && \bigl ( c^a \bigr )^{\dagger_c} \coloneq \overline{c}^a \\
		& \bigl ( \overline{c}^a \bigr )^{\dagger_C} \coloneq \overline{c}^a && \bigl ( \overline{c}^a \bigr )^{\dagger_c} \coloneq c^a \\
		& \bigl ( b^a \bigr )^{\dagger_C} \coloneq b^a && \bigl ( b^a \bigr )^{\dagger_c} \coloneq - b^a - \mathrm{g} \xi \tensor{f}{^a _b _c} \overline{c}^b c^c \\
		& \bigl ( {b^\prime}^a \bigr )^{\dagger_C} \coloneq {b^\prime}^a && \bigl ( {b^\prime}^a \bigr )^{\dagger_c} \coloneq - {b^\prime}^a \\
		& \bigl ( \partial_\mu \bigr )^{\dagger_C} \coloneq - \partial_\mu && \bigl ( \partial_\mu \bigr )^{\dagger_c} \coloneq - \partial_\mu \\
		& \bigl ( \tensor{\Gamma}{^\rho _\mu _\nu} \bigr )^{\dagger_C} \coloneq - \tensor{\Gamma}{^\rho _\mu _\nu} && \bigl ( \tensor{\Gamma}{^\rho _\mu _\nu} \bigr )^{\dagger_c} \coloneq - \tensor{\Gamma}{^\rho _\mu _\nu} \\
		& \bigl ( \imaginary \tensor{f}{^a _b _c} \bigr )^{\dagger_C} \coloneq - \imaginary \tensor{f}{^a _b _c} && \bigl ( \imaginary \tensor{f}{^a _b _c} \bigr )^{\dagger_c} \coloneq - \imaginary \tensor{f}{^a _b _c} \\
		& \bigl ( \varphi \bigr )^{\dagger_C} \coloneq \varphi && \bigl ( \varphi \bigr )^{\dagger_c} \coloneq \varphi
	\end{align}
	\end{subequations}
	}%
	Here, \({B^\prime}^\rho\) and \({b^\prime}^a\) are the shifted anti-Hermitian Lautrup--Nakanishi auxiliary fields, given via
	\begin{subequations}
	\begin{align}
		{B^\prime}^\rho & \coloneq B^\rho - \frac{\varkappa \zeta}{2} \left ( \overline{C}^\sigma \bigl ( \partial_\sigma C^\rho \bigr ) - \bigl ( \partial_\sigma \overline{C}^\rho \bigr ) C^\sigma \right )
		\intertext{and}
		{b^\prime}^a & \coloneq b^a - \frac{\mathrm{g} \xi}{2} \tensor{f}{^a _b _c} \overline{c}^b c^c \, .
	\end{align}
	\end{subequations}
	And then, finally, we introduce the total ghost conjugation \(\dagger\) as follows:
	{\allowdisplaybreaks
	\begin{subequations}
	\begin{align}
		& \bigl ( C^\rho \bigr )^\dagger \coloneq \overline{C}^\rho \\
		& \bigl ( \overline{C}^\rho \bigr )^\dagger \coloneq C^\rho \\
		& \bigl ( B^\rho \bigr )^\dagger \coloneq - B^\rho - \varkappa \zeta \left ( \overline{C}^\sigma \bigl ( \partial_\sigma C^\rho \bigr ) - \bigl ( \partial_\sigma \overline{C}^\rho \bigr ) C^\sigma \right ) \\
		& \bigl ( {B^\prime}^\rho \bigr )^\dagger \coloneq - {B^\prime}^\rho \\
		& \bigl ( c^a \bigr )^\dagger \coloneq \overline{c}^a \\
		& \bigl ( \overline{c}^a \bigr )^\dagger \coloneq c^a \\
		& \bigl ( b^a \bigr )^\dagger \coloneq - b^a - \mathrm{g} \xi \tensor{f}{^a _b _c} \overline{c}^b c^c \\
		& \bigl ( {b^\prime}^a \bigr )^\dagger \coloneq - {b^\prime}^a \\
		& \bigl ( \partial_\mu \bigr )^\dagger \coloneq - \partial_\mu \\
		& \bigl ( \tensor{\Gamma}{^\rho _\mu _\nu} \bigr )^\dagger \coloneq - \tensor{\Gamma}{^\rho _\mu _\nu} \\
		& \bigl ( \imaginary \tensor{f}{^a _b _c} \bigr )^\dagger \coloneq - \imaginary \tensor{f}{^a _b _c} \\
		& \bigl ( \varphi \bigr )^\dagger \coloneq \varphi
	\end{align}
	\end{subequations}
	}%
	In particular, the total ghost conjugation inverts simultaneously graviton-ghosts and gauge ghosts.
\end{defn}

\enter

\begin{rem} \label{rem:anomalies}
	We emphasize that the BRST setup presented in this article operates on a classical level and lifts to the corresponding Quantum Field Theory only if there are no diffeomorphism and gauge anomalies present. Specifically, it is compatible with the path integral quantization if the corresponding Slavnov--Taylor identities are satisfied. This is of course an important assumption, as it directly relates to the multiplicative renormalization of the theory, cf.\ \cite{Prinz_7,Prinz_3}. This will be studied in future work, using the proposed differential-graded renormalization Hopf algebra for (generalized) gauge theories \cite{Prinz_7,Prinz_PhD} and the BV formalism \cite{Batalin_Vilkovisky_1,Batalin_Vilkovisky_2}.
\end{rem}

\enter

\section{The diffeomorphism complex} \label{sec:diffeomorphism-brst-complex}

In this section, we study the diffeomorphism BRST operator \(P\) and its corresponding anti-operator \(\overline{P}\) together with the de Donder gauge fixing fermion \(\stigma\) and its linearized variant \(\stigma^{(1)}\).

\enter

\begin{defn} \label{defn:diffeomorphism_brst_operator}
	We define the diffeomorphism BRST operator \(P \in \mathfrak{X}_{(1,0)} \left ( \BQ \right )\) as the following odd vector field on the spacetime-matter bundle with graviton-ghost degree 1:\footnote{Here, \(\nabla^{\Sigma M}_\rho\) denotes the covariant derivative on the twisted spinor bundle, cf.\ e.g.\ \cite[Definition 2.10]{Prinz_2}.}
	\begin{equation}
	\begin{split}
		P & \coloneq \left ( \frac{1}{\zeta} \partial_\mu C_\nu + \frac{1}{\zeta} \partial_\nu C_\mu - 2 C_\rho \tensor{\Gamma}{^\rho _\mu _\nu} \right ) \frac{\partial}{\partial h_{\mu \nu}} + \varkappa C^\rho \left ( \partial_\rho C_\sigma \right ) \frac{\partial}{\partial C_\sigma} + \frac{1}{\zeta} B^\sigma \frac{\partial}{\partial \overline{C}^\sigma} \\
		& \phantom{\coloneq} + \varkappa \left ( C^\rho \bigl ( \partial_\rho A_\mu^a \bigr ) + \left ( \partial_\mu C^\rho \right ) A_\rho^a \right ) \frac{\partial}{\partial A_\mu^a} \\
		& \phantom{\coloneq} + \varkappa C^\rho \left ( \partial_\rho c^a \right ) \frac{\partial}{\partial c^a} + \varkappa C^\rho \left ( \partial_\rho \overline{c}^a \right ) \frac{\partial}{\partial \overline{c}^a} + \varkappa C^\rho \left ( \partial_\rho b^a \right ) \frac{\partial}{\partial b^a} \\
		& \phantom{\coloneq} + \varkappa C^\rho \left ( \partial_\rho \Phi \right ) \frac{\partial}{\partial \Phi} + \varkappa \left (C^\rho  \nabla^{\Sigma M}_\rho \Psi + \frac{\imaginary}{4} \left ( \partial_\mu C_\nu - \partial_\nu C_\mu \right ) e^{\mu m} e^{\nu n} \boldsymbol{\sigma}_{mn} \Psi \right ) \frac{\partial}{\partial \Psi}
	\end{split}
	\end{equation}
	Equivalently, its action on fundamental particle fields is given as follows:
	{\allowdisplaybreaks
	\begin{subequations}
	\begin{align}
		P h_{\mu \nu} & \coloneq \frac{1}{\zeta} \partial_\mu C_\nu + \frac{1}{\zeta} \partial_\nu C_\mu - 2 C_\rho \tensor{\Gamma}{^\rho _\mu _\nu} \\
		P C_\rho & \coloneq \varkappa C^\sigma \left ( \partial_\sigma C_\rho \right ) \\
		P \overline{C}^\rho & \coloneq \frac{1}{\zeta} B^\rho \\
		P B^\rho & \coloneq 0 \\
		P \eta_{\mu \nu} & \coloneq 0 \\
		P \partial_\mu & \coloneq 0 \\
		P \tensor{\Gamma}{^\rho _\mu _\nu} & \coloneq \left ( C^\sigma \bigl ( \partial_\sigma \tensor{\Gamma}{^\rho _\mu _\nu} \bigr ) + \left ( \partial_\mu C^\sigma \right ) \tensor{\Gamma}{^\rho _\sigma _\nu} + \left ( \partial_\nu C^\sigma \right ) \tensor{\Gamma}{^\rho _\mu _\sigma} - \left ( \partial_\sigma C^\rho \right ) \tensor{\Gamma}{^\sigma _\mu _\nu} + \partial_\mu \partial_\nu C^\rho \right ) \\
		P A_\mu^a & \coloneq \varkappa \left ( C^\rho \bigl ( \partial_\rho A_\mu^a \bigr ) + \left ( \partial_\mu C^\rho \right ) A_\rho^a \right ) \\
		P c^a & \coloneq \varkappa C^\rho \left ( \partial_\rho c^a \right ) \\
		P \overline{c}_a & \coloneq \varkappa C^\rho \left ( \partial_\rho \overline{c}^a \right ) \\
		P b_a & \coloneq \varkappa C^\rho \left ( \partial_\rho b^a \right ) \\
		P \delta_{ab} & \coloneq 0 \\
		P \Phi & \coloneq \varkappa C^\rho \left ( \partial_\rho \Phi \right ) \\
		P \Psi & \coloneq \varkappa \left (C^\rho  \nabla^{\Sigma M}_\rho \Psi + \frac{\imaginary}{4} \left ( \partial_\mu C_\nu - \partial_\nu C_\mu \right ) e^{\mu m} e^{\nu n} \boldsymbol{\sigma}_{mn} \Psi \right )
	\end{align}
	\end{subequations}
	}%
	We remark that the action of \(P\) on all fields \(\varphi \notin \bigl \{ h, C, \overline{C}, B, \eta \bigr \}\) is given via the geodesic Lie derivative with respect to \(C\) and rescaled via \(\varkappa\), i.e.\ \(P \varphi \equiv \varkappa \Lie_C \varphi\).\footnote{We remark that the Lie derivative of spinor fields is a non-trivial notion: This is due to the fact that the construction of the spinor bundle depends directly on the metric, which is itself affected by the Lie derivative. We use the formula of Kosmann \cite{Kosmann}, which uses the connection on the spinor bundle. It can be shown, however, that the result is indeed independent of the chosen connection. We remark that this formula can be embedded into the construction of a universal spinor bundle cf.\ \cite{Mueller_Nowaczyk}.}
\end{defn}

\enter

\begin{defn} \label{defn:diffeomorphism_anti-brst_operator}
	Given the situation of \defnref{defn:diffeomorphism_brst_operator}, we additionally define the diffeomorphism anti-BRST operator \(\overline{P} \in \mathfrak{X}_{(-1,0)} \left ( \BQ \right )\) as the following odd vector field on the spacetime-matter bundle with graviton-ghost degree -1:
	\begin{subequations} \label{eqns:diffeomorphism_anti-brst_operator}
	\begin{align}
		\overline{P} & \coloneq \eval{P}_{C \rightsquigarrow \overline{C}} \, ,
		\intertext{together with the following additional changes}
		\overline{P} C_\rho & \coloneq - \frac{1}{\zeta} B_\rho + \varkappa \overline{C}^\sigma \left ( \partial_\sigma C_\rho \right ) + \varkappa \bigl ( \partial_\rho \overline{C}^\sigma \bigr ) C_\sigma \, , \\
		\overline{P} \overline{C}^\rho & \coloneq \varkappa \overline{C}^\sigma \bigl ( \partial_\sigma \overline{C}^\rho \bigr )
		\intertext{and}
		\overline{P} B^\rho & \coloneq \varkappa \overline{C}^\sigma \left ( \partial_\sigma B^\rho \right ) - \varkappa \bigl ( \partial_\sigma \overline{C}^\rho \bigr ) B^\sigma \, .
	\end{align}
	\end{subequations}
	Specifically, the diffeomorphism anti-BRST operator \(\overline{P}\) is related to the diffeomorphism BRST operator \(P\) via the diffeomorphism-ghost conjugation \(\dagger_C\) from \defnref{defn:ghost-conjugation}, as will be shown in \lemref{lem:relation_BRST_anti-BRST_operators}.
\end{defn}

\enter

\begin{prop} \label{prop:p-cohomological-vector-field}
	Given the situation of \defnref{defn:diffeomorphism_brst_operator}, we have
	\begin{equation}
		\commutatorbig{P}{P} \equiv 2 P^2 \equiv 0 \, ,
	\end{equation}
	i.e.\ \(P\) is a cohomological vector field with respect to the graviton-ghost degree.
\end{prop}

\begin{proof}
	This follows immediately after a short calculation using the Jacobi identity.
\end{proof}

\enter

\begin{col} \label{col:anti-diffeomorphism_brst_operator}
	Given the situation of \defnref{defn:diffeomorphism_anti-brst_operator}, we have
	\begin{align}
		\commutatorbig{\overline{P}}{\overline{P}} & \equiv 2 \overline{P}^2 \equiv 0
		\intertext{and}
		\commutatorbig{P}{\overline{P}} & \equiv P \circ \overline{P} + \overline{P} \circ P \equiv 0 \, ,
	\end{align}
	i.e.\ \(\overline{P}\) is a homological vector field with respect to the graviton-ghost degree that anticommutes with \(P\).
\end{col}

\begin{proof}
	This statement can be shown analogously to \propref{prop:p-cohomological-vector-field}.
\end{proof}

\enter

\begin{lem} \label{lem:p_tensor_densities}
	Let \(\mathfrak{f} \in \CQ^{\{w\}, (i,j)}\) be a non-constant covariant functional in ghost-bidegree \((i,j) \in \mathbb{Z}^2\) and tensor density weight \(w \in \mathbb{R}\). Then if \(\mathfrak{f}\) does not involve graviton-ghosts \(C\) and graviton-antighosts \(\overline{C}\), the following statements are equivalent:\footnote{Generally, the statement does not hold for the gravitational Koszul--Tate resolution of the complex: Specifically, we have \(P \overline{C}^\rho = B^\rho\) and \(\overline{P} C_\rho = - B_\rho / \zeta + \varkappa \overline{C}^\sigma \bigl ( \partial_\sigma C_\rho \bigr ) + \varkappa \bigl ( \partial_\rho \overline{C}^\sigma \bigr ) C_\sigma\), and thus the parts involving the gravitational-Lautrup--Nakanishi field do not turn into total derivatives. In particular, for statement 2 we could furthermore allow graviton-ghosts, whereas for statement 3 we could furthermore allow graviton-antighosts.}
	\begin{enumerate}
		\item \(w = 1\)
		\item \(P \mathfrak{f} \simeq_\textup{TD} 0\)
		\item \(\overline{P} \mathfrak{f} \simeq_\textup{TD} 0\)
	\end{enumerate}
	In particular, every functional can be modified uniquely to satisfy condition 1., where, \(\simeq_\textup{TD}\) means equality modulo total derivatives.
\end{lem}

\begin{proof}
	The equivalence between statements 1 and 2 follows directly from the calculation
	\begin{equation}
	\begin{split}
		P \mathfrak{f} & = \Lie_C \mathfrak{f} \\
		& = C^\rho \left ( \partial_\rho \mathfrak{f} \right ) + w \left ( \partial_\rho C^\rho \right ) \mathfrak{f} \\
		& = \partial_\rho \left ( C^\rho \mathfrak{f} \right ) + \left ( w - 1 \right ) \left ( \partial_\rho C^\rho \right ) \mathfrak{f} \, ,
		\end{split}
	\end{equation}
	which is a total derivative if and only if \(w = 1\). The equivalence between statements 2 and 3 follows directly from \lemref{lem:relation_BRST_anti-BRST_operators}, which states that \(\overline{P} \equiv P^{\dagger_C}\), where \(\dagger_C\) is the graviton-ghost conjugation. This then also concludes the equivalence between 1 and 3. Finally, the last claim that every covariant functional can be modified uniquely to satisfy condition 1 is due to the following argument: Let \(\mathfrak{f} \in \CQ^{\{w\}, (i,j)}\) be a non-constant covariant functional of tensor density weight \(w\). Then,
	\begin{equation}
		\tilde{\mathfrak{f}} \coloneq \sqrt{- \dt{g}}^{\left ( 1 - w \right )} \mathfrak{f}
	\end{equation}	
	is a non-constant covariant functional of tensor density weight \(w = 1\).
\end{proof}

\enter

\begin{prop} \label{prop:de_donder_gauge_fixing_fermion}
	The Quantum General Relativity gauge fixing Lagrange density and its accompanying ghost Lagrange density
	\begin{equation}
	\begin{split}
		\mathcal{L}_\textup{GR-GF} + \mathcal{L}_\textup{GR-Ghost} & = - \frac{1}{4 \varkappa^2 \zeta}  g^{\mu \nu} \deDonder_\mu \deDonder_\nu \dif V_g \\ & \phantom{=} - \frac{1}{2 \zeta} g_{\rho \sigma} g^{\mu \nu} \overline{C}^\rho \left ( \partial_\mu \partial_\nu C^\sigma \right ) \dif V_g \\ & \phantom{=} - \frac{1}{2} \overline{C}^\rho \left ( \left ( \partial_\sigma \deDonder_\rho \right ) C^\sigma + \deDonder_\sigma \left ( \partial_\rho C^\sigma \right ) \right ) \dif V_g
	\end{split}
	\end{equation}
	for the de Donder gauge fixing functional \(\deDonder_\mu \coloneq g^{\rho \sigma} \Gamma_{\mu \rho \sigma}\) can be obtained from the following gauge fixing fermion \(\stigma \in \CQ^{(-1,0)}\)
	\begin{equation}
		\stigma \coloneq \frac{1}{2} \overline{C}^\rho \left ( \frac{1}{\varkappa} \deDonder_\rho + \frac{1}{2} B_\rho \right ) \dif V_g
	\end{equation}
	via \(P \stigma\).
\end{prop}

\begin{proof}
	The claimed statement follows directly from the calculations
	\begin{subequations}
	\begin{align}
		\begin{split}
			P \stigma & = \frac{1}{2 \zeta} B^\rho \left ( \frac{1}{\varkappa} \deDonder_\rho + \frac{1}{2} B_\rho \right ) \dif V_g  - \frac{1}{2 \varkappa} \overline{C}^\rho \left ( P \deDonder_\rho \right ) \dif V_g \\ & \phantom{=} - \frac{1}{2} \overline{C}^\rho \left ( \frac{1}{\varkappa} \deDonder_\rho + \frac{1}{2} B_\rho \right ) \left ( P \dif V_g \right ) \label{eqn:p_stigma}
		\end{split}
		\intertext{with}
		\begin{split}
			P \deDonder_\rho & = P \left ( g^{\mu \nu} \Gamma_{\rho \mu \nu} \right ) \\
			& = C^\sigma \left ( \partial_\sigma \deDonder_\rho \right ) + \left ( \partial_\rho C^\sigma \right ) \deDonder_\sigma + g_{\rho \sigma} g^{\mu \nu} \left ( \partial_\mu \partial_\nu C^\sigma \right )
		\end{split}
		\intertext{along with the total derivative}
		P \dif V_g & = \partial_\rho \left ( C^\rho \dif V_g \right )
		\intertext{and then finally eliminating the Lautrup--Nakanishi auxiliary field \(B^\rho\) by inserting its equation of motion}
		\operatorname{EoM} \left ( B_\rho \right ) & = - \frac{1}{\varkappa} \deDonder_\rho \, ,
	\intertext{which are obtained as usual via an Euler--Lagrange variation of \eqnref{eqn:p_stigma}, i.e.\ by solving}
	0 & \overset{!}{=} \left ( \left ( \frac{\partial}{\partial B_\rho} \right ) - \partial_\mu \left ( \frac{\partial}{\partial \left ( \partial_\mu B_\rho \right )} \right ) \right ) P \stigma \, ,
	\end{align}
	\end{subequations}
	where the second term vanishes identically, as \(B_\rho\) is a Lagrange multiplier and thus has no kinetic term.
\end{proof}

\enter

\begin{col} \label{col:linearized_de_donder_gauge_fixing_fermion}
	Given the situation of \propref{prop:de_donder_gauge_fixing_fermion}. Then the linearized de Donder gauge fixing and ghost Lagrange densities read
	\begin{equation}
	\begin{split}
		\mathcal{L}_\textup{GR-GF} + \mathcal{L}_\textup{GR-Ghost} & = - \frac{1}{4 \varkappa^2 \zeta}  \eta^{\mu \nu} \deDonder^{(1)}_\mu \deDonder^{(1)}_\nu \dif V_\eta \\ & \phantom{=} - \frac{1}{2 \zeta} \eta^{\mu \nu} \overline{C}^\rho \left ( \partial_\mu \partial_\nu C_\rho \right ) \dif V_\eta \\ & \phantom{=} - \frac{1}{2} \eta^{\mu \nu} \overline{C}^\rho \left ( \partial_\rho \bigl ( \Gamma_{\sigma \mu \nu} C^\sigma \bigr ) - 2 \partial_\mu \bigl ( \Gamma_{\sigma \nu \rho} C^\sigma \bigr ) \right ) \dif V_\eta
	\end{split}
	\end{equation}
	with the linearized de Donder gauge fixing functional \(\deDonder^{(1)}_\mu \coloneq \eta^{\rho \sigma} \Gamma_{\mu \rho \sigma}\). They can be obtained from the following gauge fixing fermion \(\stigma^{(1)} \in \CQ^{(-1,0)}\)
	\begin{equation}
		\stigma^{(1)} \coloneq \frac{1}{2} \overline{C}^\rho \left ( \frac{1}{\varkappa} \deDonder^{(1)}_\rho + \frac{1}{2} B_\rho \right ) \dif V_\eta \label{eqn:linearized_de_donder_gauge_fixing_fermion}
	\end{equation}
	via \(P \stigma^{(1)}\).
\end{col}

\begin{proof}
	This can be shown analogously to the proof of \propref{prop:de_donder_gauge_fixing_fermion}.
\end{proof}

\enter

\begin{rem}
	In the following, we will only use the linearized de Donder gauge fixing and ghost Lagrange densities from \colref{col:linearized_de_donder_gauge_fixing_fermion}. The reason is that the perturbative expansion becomes simpler if the gauge fixing functional does only contribute to the propagator. Nevertheless, the complete de Donder gauge fixing can also be useful, as it does not depend on the choice of a background metric.
\end{rem}

\enter

\section{The gauge complex} \label{sec:gauge-brst-complex}

In this section, we study the gauge BRST operator \(Q\) and its corresponding anti-operator \(\overline{Q}\) together with the Lorenz gauge fixing fermion \(\digamma\) and its density variant \(\digamma \! \! _{\{ 1 \}}\).

\enter

\begin{defn} \label{defn:gauge_brst_operator}
	We define the gauge BRST operator \(Q \in \mathfrak{X}_{(0,1)} \left ( \BQ \right )\) as the following odd vector field on the spacetime-matter bundle with gauge ghost degree 1:
	\begin{equation}
	\begin{split}
		Q & \coloneq \left ( \frac{1}{\xi} \partial_\mu c^a + \mathrm{g} \tensor{f}{^a _b _c} c^b A_\mu^c \right ) \frac{\partial}{\partial A_\mu^a} + \frac{\mathrm{g}}{2} \tensor{f}{^a _b _c} c^b c^c \frac{\partial}{\partial c^a} + \frac{1}{\xi} b_a \frac{\partial}{\partial \overline{c}_a} \\ & \phantom{\coloneq} + \mathrm{g} c^a \left ( \mathfrak{H}_a \cdot \Phi \right ) \frac{\partial}{\partial \Phi} + \mathrm{g} c^a \left ( \mathfrak{S}_a \cdot \Psi \right ) \frac{\partial}{\partial \Psi}
	\end{split}
	\end{equation}
	Equivalently, its action on fundamental particle fields is given as follows:
	{\allowdisplaybreaks
	\begin{subequations}
	\begin{align}
		Q A_\mu^a & \coloneq \frac{1}{\xi} \partial_\mu c^a + \mathrm{g} \tensor{f}{^a _b _c} c^b A_\mu^c \\
		Q c^a & \coloneq \frac{\mathrm{g}}{2} \tensor{f}{^a _b _c} c^b c^c \\
		Q \overline{c}_a & \coloneq \frac{1}{\xi} b_a \\
		Q b_a & \coloneq 0 \\
		Q \delta_{ab} & \coloneq 0 \\
		Q h_{\mu \nu} & \coloneq 0 \\
		Q C_\rho & \coloneq 0 \\
		Q \overline{C}^\rho & \coloneq 0 \\
		Q B^\rho & \coloneq 0 \\
		Q \eta_{\mu \nu} & \coloneq 0 \\
		Q \partial_\mu & \coloneq 0 \\
		Q \tensor{\Gamma}{^\rho _\mu _\nu} & \coloneq 0 \\
		Q \Phi & \coloneq \mathrm{g} c^a \left ( \mathfrak{H}_a \cdot \Phi \right ) \\
		Q \Psi & \coloneq \mathrm{g} c^a \left ( \mathfrak{S}_a \cdot \Psi \right )
	\end{align}
	\end{subequations}
	}%
	We remark that the action of \(Q\) on all fields \(\varphi \notin \bigl \{ A, c, \overline{c}, b, \delta \bigr \}\) is given via the gauge Lie derivative with respect to \(c\) and rescaled via \(\mathrm{g}\), i.e.\ \(Q \varphi \equiv \mathrm{g} \lie_c \varphi\).
\end{defn}

\enter

\begin{defn} \label{defn:gauge_anti-brst_operator}
	Given the situation of \defnref{defn:gauge_brst_operator}, we additionally define the gauge anti-BRST operator \(\overline{Q} \in \mathfrak{X}_{(0,-1)} \left ( \BQ \right )\) as the following odd vector field on the spacetime-matter bundle with gauge ghost degree -1:
	\begin{subequations} \label{eqns:gauge_anti-brst_operator}
	\begin{align}
		\overline{Q} & \coloneq \eval{Q}_{c \rightsquigarrow \overline{c}} \, ,
		\intertext{together with the following additional changes}
		\overline{Q} c^a & \coloneq - \frac{1}{\xi} b^a + \mathrm{g} \tensor{f}{^a _b _c} \overline{c}^b c^c \, , \\
		\overline{Q} \overline{c}_a & \coloneq \frac{\mathrm{g}}{2} \tensor{f}{_a _b _c} \overline{c}^b \overline{c}^c
		\intertext{and}
		\overline{Q} b_a & \coloneq \mathrm{g} \tensor{f}{_a _b _c} \overline{c}^b b^c \, .
	\end{align}
	\end{subequations}
	Specifically, the gauge anti-BRST operator \(\overline{Q}\) is related to the gauge BRST operator \(Q\) via the gauge ghost conjugation \(\dagger_c\) from \defnref{defn:ghost-conjugation}, as will be shown in \lemref{lem:relation_BRST_anti-BRST_operators}.
\end{defn}

\enter

\begin{prop} \label{prop:q-cohomological-vector-field}
	Given the situation of \defnref{defn:gauge_brst_operator}, we have
	\begin{equation}
		\commutatorbig{Q}{Q} \equiv 2 Q^2 \equiv 0 \, ,
	\end{equation}
	i.e.\ \(Q\) is a cohomological vector field with respect to the gauge ghost degree.
\end{prop}

\begin{proof}
	This follows immediately after a short calculation using the Jacobi identity.
\end{proof}

\enter

\begin{col} \label{col:anti-gauge_brst_operator}
	Given the situation of \defnref{defn:gauge_anti-brst_operator}, we have
	\begin{align}
		\commutatorbig{\overline{Q}}{\overline{Q}} & \equiv 2 \overline{Q}^2 \equiv 0
		\intertext{and}
		\commutatorbig{Q}{\overline{Q}} & \equiv Q \circ \overline{Q} + \overline{Q} \circ Q \equiv 0 \, ,
	\end{align}
	i.e.\ \(\overline{Q}\) is a homological vector field with respect to the gauge ghost degree that anticommutes with \(Q\).
\end{col}

\begin{proof}
	This statement can be shown analogously to \propref{prop:q-cohomological-vector-field}.
\end{proof}

\enter

\begin{prop} \label{prop:lorenz_gauge_fixing_fermion}
	The Quantum Yang--Mills theory gauge fixing Lagrange density and its accompanying ghost Lagrange density
	\begin{equation} \label{eqn:qym-linearized-lorenz-gf-ghost}
	\begin{split}
		\mathcal{L}_{\textup{QYM-GF}} + \mathcal{L}_{\textup{QYM-Ghost}} & = - \frac{1}{2 \mathrm{g}^2 \xi}  \delta_{a b} \linLorenz^a \linLorenz^b \dif V_\eta \\ & \phantom{=} - \frac{1}{\xi} \eta^{\mu \nu} \overline{c}_a \left ( \partial_\mu \partial_\nu c^a \right ) \dif V_\eta \\ & \phantom{=} - \mathrm{g} \eta^{\mu \nu} \tensor{f}{^a _b _c} \overline{c}_a \left ( \partial_\mu \bigl ( c^b A^c_\nu \bigr ) \right ) \dif V_\eta
	\end{split}
	\end{equation}
	for the linearized Lorenz gauge fixing functional \(\linLorenz^a \coloneq \mathrm{g} \eta^{\mu \nu} \left ( \partial_\mu A_\nu^a \right )\) can be obtained from the following gauge fixing fermion \(\digamma \in \CQ^{\{ 0 \}, (0,-1)}\)
	\begin{equation}
		\digamma \coloneq \overline{c}_a \left ( \frac{1}{\mathrm{g}} \linLorenz^a + \frac{1}{2} b^a \right ) \dif V_\eta
	\end{equation}
	via \(Q \digamma\).\footnote{The linearized Lorenz gauge fixing condition for the electroweak sector needs the following adjustments: In the following, \(A\) denotes the photon, \(Z\) the \(Z\)-boson, \(W^\pm\) the \(W^\pm\)-bosons and \(\phi_Z\), \(\phi_{W^+}\) and \(\phi_{W^-}\) the corresponding Goldstone bosons. In addition, we have the corresponding coupling constants \(\mathrm{e}_A\), \(\mathrm{e}_Z\) and \(\mathrm{e}_W\) and gauge fixing parameters \(\xi_A\), \(\xi_Z\) and \(\xi_W\). Then, the corresponding gauge fixing functionals are given via: \begin{subequations} \begin{align} \linLorenz_A & \coloneq \mathrm{e} \eta^{\mu \nu} \bigl ( \partial_\mu A_\nu \bigr ) \, , \\ \linLorenz_Z & \coloneq \mathrm{e}_Z \left ( \eta^{\mu \nu} \bigl ( \partial_\mu Z_\nu \bigr ) + \xi_Z m_Z \phi_Z \right ) \, , \\ \linLorenz_{W^\pm} & \coloneq \mathrm{e}_W \left ( \eta^{\mu \nu} \bigl ( \partial_\mu W^{\pm}_\nu \bigr ) \pm \imaginary \xi_W m_W \phi_{W^\pm} \right ) \intertext{and with corresponding gauge fixing Lagrange density} \mathcal{L}_{\textup{EW-GF}} & \coloneq - \frac{1}{2} \biggl ( \frac{1}{\mathrm{e}_A^2 \xi_A} \linLorenz_A^2 + \frac{1}{\mathrm{e}_Z^2 \xi_Z} \linLorenz_Z^2 + \frac{1}{\mathrm{e}_W^2 \xi_W} \left ( \linLorenz_{W^+} \linLorenz_{W^-} \right ) \biggr ) \, . \end{align} \end{subequations}}
\end{prop}

\begin{proof}
	The claimed statement follows directly from the calculations
	\begin{subequations}
	\begin{align}
		Q \digamma & = \frac{1}{\xi} b_a \left ( \frac{1}{\mathrm{g}} \linLorenz^a + \frac{1}{2} b^a \right ) \dif V_\eta - \frac{1}{\mathrm{g}} \overline{c}_a \left ( Q \linLorenz^a \right ) \dif V_\eta \label{eqn:q_digamma_eta}
		\intertext{with}
		\begin{split}
		Q \linLorenz^a & = \mathrm{g} \eta^{\mu \nu} \partial_\mu \left ( Q A_\nu^a \right ) \\
		& = \mathrm{g} \eta^{\mu \nu} \partial_\mu \left ( \frac{1}{\xi} \partial_\nu c^a + \mathrm{g} \tensor{f}{^a _b _c} c^b A_\nu^c \right )
		\end{split}
		\intertext{and then finally eliminating the Lautrup--Nakanishi auxiliary field \(b_a\) by inserting its equation of motion}
		\operatorname{EoM} \left ( b_a \right ) & = - \frac{1}{\mathrm{g}} \linLorenz^a \, ,
	\intertext{which are obtained as usual via an Euler--Lagrange variation of \eqnref{eqn:q_digamma_eta}, i.e.\ by solving}
	0 & \overset{!}{=} \left ( \left ( \frac{\partial}{\partial b_a} \right ) - \partial_\mu \left ( \frac{\partial}{\partial \left ( \partial_\mu b_a \right )} \right ) \right ) Q \digamma \, ,
	\end{align}
	\end{subequations}
	where the second term vanishes identically, as \(b_a\) is a Lagrange multiplier and thus has no kinetic term.
\end{proof}

\enter

\begin{col} \label{col:covariant_lorenz_gauge_fixing_fermion}
	Given the situation of \propref{prop:lorenz_gauge_fixing_fermion}. Then the covariant Lorenz gauge fixing and ghost Lagrange densities read
	\begin{equation} \label{eqn:qym-covariant-lorenz-gf-ghost}
	\begin{split}
		\mathcal{L}_{\textup{QYM-GF}} + \mathcal{L}_{\textup{QYM-Ghost}} & = - \frac{1}{2 \mathrm{g}^2 \xi}  \delta_{a b} \Lorenz^a \Lorenz^b \dif V_g \\ & \phantom{=} - \frac{1}{\xi} g^{\mu \nu} \overline{c}_a \left ( \nabla^{TM}_\mu \left ( \partial_\nu c^a \right ) \right ) \dif V_g \\ & \phantom{\coloneq} - \mathrm{g} g^{\mu \nu} \tensor{f}{^a _b _c} \overline{c}_a \left ( \nabla^{TM}_\mu \bigl ( c^b A^c_\nu \bigr ) \right ) \dif V_g \, ,
	\end{split}
	\end{equation}
	with the covariant Lorenz gauge fixing functional \(\Lorenz^a \coloneq \mathrm{g} g^{\mu \nu} \bigl ( \nabla^{TM}_\mu A_\nu^a \bigr )\). They can be obtained from the following gauge fixing fermion \(\digamma \! \! _{\{ 1 \}} \in \CQ^{\{ 1 \}, (0,-1)}\)
	\begin{equation}
		\digamma \! \! _{\{ 1 \}} \coloneq \overline{c}_a \left ( \frac{1}{\mathrm{g}} \Lorenz^a + \frac{1}{2} b^a \right ) \dif V_g \label{eqn:covariant_lorenz_gauge_fixing_fermion}
	\end{equation}
	via \(Q \digamma \! \! _{\{ 1 \}}\).\footnote{The covariant Lorenz gauge fixing condition for the electroweak sector needs the following adjustments: In the following, \(A\) denotes the photon, \(Z\) the \(Z\)-boson, \(W^\pm\) the \(W^\pm\)-bosons and \(\phi_Z\), \(\phi_{W^+}\) and \(\phi_{W^-}\) the corresponding Goldstone bosons. In addition, we have the corresponding coupling constants \(\mathrm{e}_A\), \(\mathrm{e}_Z\) and \(\mathrm{e}_W\) and gauge fixing parameters \(\xi_A\), \(\xi_Z\) and \(\xi_W\). Then, the corresponding gauge fixing functionals are given via: \begin{subequations} \begin{align} \Lorenz_A & \coloneq \mathrm{e} g^{\mu \nu} \bigl ( \nabla^{TM}_\mu A_\nu \bigr ) \, , \\ \Lorenz_Z & \coloneq \mathrm{e}_Z \left ( g^{\mu \nu} \bigl ( \nabla^{TM}_\mu Z_\nu \bigr ) + \xi_Z m_Z \phi_Z \right ) \, , \\ \Lorenz_{W^\pm} & \coloneq \mathrm{e}_W \left ( g^{\mu \nu} \bigl ( \nabla^{TM}_\mu W^{\pm}_\nu \bigr ) \pm \imaginary \xi_W m_W \phi_{W^\pm} \right ) \intertext{and with corresponding gauge fixing Lagrange density} \mathcal{L}_{\textup{EW-GF}} & \coloneq - \frac{1}{2} \biggl ( \frac{1}{\mathrm{e}_A^2 \xi_A} \Lorenz_A^2 + \frac{1}{\mathrm{e}_Z^2 \xi_Z} \Lorenz_Z^2 + \frac{1}{\mathrm{e}_W^2 \xi_W} \left ( \Lorenz_{W^+} \Lorenz_{W^-} \right ) \biggr ) \, . \end{align} \end{subequations}}
\end{col}

\begin{proof}
	This can be shown analogously to the proof of \propref{prop:lorenz_gauge_fixing_fermion}.
\end{proof}

\enter

\section{The diffeomorphism-gauge BRST double complex} \label{sec:diffeomorphism-gauge-brst-double-complex}

In this section, we show that the two BRST operators \(P\) and \(Q\) anticommute and thus give rise to the \emph{total BRST operator} as the sum \(D \coloneq P + Q\). Additionally, we show that each gauge theory gauge fixing fermion can be modified uniquely to become a tensor density of weight \(w = 1\). This is a useful choice, as then the graviton-ghosts decouple from matter of the Standard Model. Finally, we introduce the \emph{total gauge fixing fermion} as the sum \(\Upsilon \coloneq \stigma^{(1)} + \digamma \! \! _{\{ 1 \}}\), where \(\stigma^{(1)}\) is the gauge fixing fermion corresponding to the linearized de Donder gauge fixing and \(\digamma \! \! _{\{ 1 \}}\) is the gauge fixing fermion corresponding to the covariant Lorenz gauge fixing. This setup allows us to create the complete gauge fixing and ghost Lagrange densities of (effective) Quantum General Relativity coupled to the Standard Model via \(D \Upsilon\). We refer to \cite{Cattaneo_Schiavina} for the case of the BV-BFV formalism.

\enter

\begin{thm} \label{thm:total_brst_operator}
	Given the two BRST operators \(P \in \mathfrak{X}_{(1,0)} \left ( \BQ \right )\) and \(Q \in \mathfrak{X}_{(0,1)} \left ( \BQ \right )\) from \defnref{defn:diffeomorphism_brst_operator} and \defnref{defn:gauge_brst_operator}, respectively. Then we have
	\begin{align}
		\commutatorbig{P}{Q} & \equiv P \circ Q + Q \circ P \equiv 0 \, ,
		\intertext{i.e.\ their sum}
		D & \coloneq P + Q
	\intertext{is also a cohomological vector field with respect to the total ghost degree, and thus satisfying}
	\commutatorbig{D}{D} & \equiv 2 D^2 \equiv 0 \, .
	\end{align}
	We call \(D \in \mathfrak{X}_{(1)} \left ( \BQ \right )\) the total BRST operator.
\end{thm}

\begin{proof}
	We show the first statement by an explicit calculation:
	\begin{equation}
	\begin{split}
		P \circ Q & = \varkappa \biggl ( \frac{1}{\xi} \left ( \partial_\mu C^\rho \right ) \left ( \partial_\rho c^a \right ) + \frac{1}{\xi} C^\rho \left ( \partial_\mu \partial_\rho c^a \right ) + \mathrm{g} \tensor{f}{^a _b _c} C^\rho \bigl ( \partial_\rho c^b \bigr ) A_\mu^c \\
		& \phantom{= \varkappa \biggl (} + \mathrm{g} \tensor{f}{^a _b _c} C^\rho c^b \bigl ( \partial_\rho A_\mu^c \bigr ) + \mathrm{g} \tensor{f}{^a _b _c} \left ( \partial_\mu C^\rho \right ) c^b A_\rho^c \biggr ) \frac{\partial}{\partial A_\mu^a} \\
		& \phantom{=} + \frac{\varkappa \mathrm{g}}{2} \tensor{f}{^a _b _c} C^\rho \biggl ( \bigl ( \partial_\rho c^b \bigr ) c^c + c^b \bigl ( \partial_\rho c^c \bigr ) \biggr ) \frac{\partial}{\partial c^a} + \frac{\varkappa}{\xi} C^\rho \left ( \partial_\rho b^a \right ) \frac{\partial}{\partial \overline{c}_a} \\
		& \phantom{=} + \varkappa \mathrm{g} C^\rho c^a \mathfrak{H}_a \cdot \left ( \partial_\rho \Phi \right ) \frac{\partial}{\partial \Phi} \\
		& \phantom{=} + \varkappa \mathrm{g} C^\rho c^a \mathfrak{S}_a \cdot \left ( \nabla^{\Sigma M}_\rho \Psi + \frac{\imaginary}{4} \left ( \partial_\mu X_\nu - \partial_\nu X_\mu \right ) e^{\mu m} e^{\nu n} \left ( \boldsymbol{\sigma}_{mn} \cdot \Psi \right ) \right ) \frac{\partial}{\partial \Psi} \\
		& = - Q \circ P \, ,
	\end{split}
	\end{equation}
	where we have used \(C^\rho c^a \equiv - c^a C^\rho\) and \(\commutatorbig{\mathfrak{S}_a}{\boldsymbol{\sigma}_{mn}} \equiv 0\). Then, the second statement follows immediately by \propsaref{prop:p-cohomological-vector-field}{prop:q-cohomological-vector-field}.
\end{proof}

\enter

\begin{col} \label{col:total_anti-brst_operator}
	Given the two anti-BRST operators \(\overline{P} \in \mathfrak{X}_{(-1,0)} \left ( \BQ \right )\) and \(\overline{Q} \in \mathfrak{X}_{(0,-1)} \left ( \BQ \right )\) from \defnref{defn:diffeomorphism_anti-brst_operator} and \defnref{defn:gauge_anti-brst_operator}, respectively. Then we have
	\begin{align}
		\commutatorbig{\overline{P}}{\overline{Q}} & \equiv \overline{P} \circ \overline{Q} + \overline{Q} \circ \overline{P} \equiv 0 \, ,
		\intertext{i.e.\ their sum}
		\overline{D} & \coloneq \overline{P} + \overline{Q}
		\intertext{is also a homological vector field with respect to the total ghost degree, and thus satisfying}
		\commutatorbig{\overline{D}}{\overline{D}} & \equiv 2 \overline{D}^2 \equiv 0 \, .
	\end{align}
	We call \(\overline{D} \in \mathfrak{X}_{(-1)} \left ( \BQ \right )\) the total anti-BRST operator.\footnote{Specifically, the total anti-BRST operator \(\overline{D}\) is related to the total BRST operator \(D\) via the total ghost conjugation \(\dagger\) from \defnref{defn:ghost-conjugation}, as will be shown in \lemref{lem:relation_BRST_anti-BRST_operators}.} Furthermore, given the three BRST operators \(P \in \mathfrak{X}_{(1,0)} \left ( \BQ \right )\), \(Q \in \mathfrak{X}_{(0,1)} \left ( \BQ \right )\) and \(D \in \mathfrak{X}_{(1)} \left ( \BQ \right )\) from \defnref{defn:diffeomorphism_brst_operator}, \defnref{defn:gauge_brst_operator} and \thmref{thm:total_brst_operator}, respectively. Then we find that all BRST and anti-BRST operators mutually anticommute.
\end{col}

\begin{proof}
	This statement can be shown analogously to \thmref{thm:total_brst_operator}, using additionally \colsaref{col:anti-diffeomorphism_brst_operator}{col:anti-gauge_brst_operator}.
\end{proof}

\enter

\begin{rem}
	We emphasize that the Hermitian ghost conjugation manifests the specific symmetry of degree-inversion, which is particular for graded superfunctions, such as functionals of fields \(\CQ\). In particular, in \lemref{lem:relation_BRST_anti-BRST_operators}, we will show that the BRST operators relate to their corresponding anti-variants via ghost conjugation. Additionally, we have seen in \colssaref{col:anti-diffeomorphism_brst_operator}{col:anti-gauge_brst_operator}{col:total_anti-brst_operator} that the corresponding Hodge--Laplace operators vanish: Geometrically, this comes from the fact that the BRST cocomplex splits into a Chevalley--Eilenberg part and a Koszul--Tate resolution. The BRST operator, as well as the anti-BRST operator, individually keep both complexes separate: For the BRST cocomplex the Chevalley--Eilenberg part is given via particle fields and ghost fields, whereas the Koszul--Tate resolution is given via antighost fields and the Lautrup--Nakanishi auxiliary field. Crucially, for the anti-BRST symmetry ghosts and antighosts switch their roles, together with a shift for the Lautrup--Nakanishi auxiliary field. Thus, composing the anti-BRST operator with the BRST operator mixes the Chevalley--Eilenberg and the Koszul--Tate complexes, which ultimately leads to the anticommutativity property of the BRST operator with its anti-BRST operator. Specifically, this differs from the de Rham codifferential, which is obtained analytically via an \(L^2\)-product.
\end{rem}

\enter

\begin{thm} \label{thm:no-couplings-grav-ghost-matter-sm}
	Let \(\mathcal{L}_\textup{SM}\) be the Standard Model Lagrange density with Yang--Mills gauge fixing fermion \(\chi\) and Lagrange density \(\mathcal{L}_{\textup{QYM-GF}} \coloneq Q \chi\). Then the following statements are equivalent:
	\begin{enumerate}
		\item Graviton-ghosts decouple from gauge bosons, gauge ghosts and matter particles
		\item \(\mathcal{L}_\textup{SM}\) is a covariant tensor density of weight \(w = 1\)
		\item \(\mathcal{L}_{\textup{QYM-GF}}\) is a covariant tensor density of weight \(w = 1\)
		\item \(\chi\) is a covariant tensor density of weight \(w = 1\)
	\end{enumerate}
	Specifically, the case of the Lorenz gauge fixing condition is given in \colref{col:covariant_lorenz_gauge_fixing_fermion}.
\end{thm}

\begin{proof}
	The equivalence between statements 1 and 2 follows from \lemref{lem:p_tensor_densities}: We know that \(P \mathcal{L}_\textup{SM} \simeq_\text{TD} 0\) if and only if \(\mathcal{L}_\textup{SM}\) is a tensor density of weight \(w = 1\). For the equivalence between statements 2 and 3, we notice that \(\mathcal{L}_{\textup{QYM-GF}}\) is the only part of \(\mathcal{L}_\textup{SM}\) where the tensor density weight is not fixed to \(w = 1\) by its requirement to be diffeomorphism invariant. Finally, the equivalence between statements 3 and 4 follows from \thmref{thm:total_brst_operator}: We have that
	\begin{subequations}
	\begin{align}
		\left ( P \circ Q \right ) & = - \left ( Q \circ P \right ) \, ,
		\intertext{and thus the equivalence}
		\left ( P \circ Q \right ) \chi \simeq_\text{TD} 0 & \iff P \chi \simeq_\text{TD} 0 \, ,
	\end{align}
	\end{subequations}
	which also concludes the equivalence between any of the mentioned statements.
\end{proof}

\enter

\begin{rem} \label{rem:use-covariant-YM-gauge-fixing-fermion}
	\thmref{thm:no-couplings-grav-ghost-matter-sm} motivates us to use the covariant Lorenz gauge fixing condition from \colref{col:covariant_lorenz_gauge_fixing_fermion} rather than the linearized version from \propref{prop:lorenz_gauge_fixing_fermion}: This avoids additional couplings between graviton-ghosts, gauge bosons and gauge ghosts at the cost of having additional vertex Feynman rules from the coupling of gravitons to the gauge boson gauge fixing condition. Nevertheless, it seems more convenient, if the gauge fixing condition breaks only gauge invariance and not also diffeomorphism invariance, e.g.\ for the \emph{total gauge fixing fermion} of the following \thmref{thm:total_gauge_fixing_fermion}.
\end{rem}

\enter

\begin{thm} \label{thm:total_gauge_fixing_fermion}
	Given the total BRST operator \(D\) from \thmref{thm:total_brst_operator}, then we define the total gauge fixing fermion \(\Upsilon \in \CQ^{(-1)}\) as the following sum
	\begin{equation}
		\Upsilon \coloneq \stigma^{(1)} + \digamma \! \! _{\{ 1 \}} \, ,
	\end{equation}
	where \(\stigma^{(1)}\) is the linearized de Donder gauge fixing fermion and \(\digamma \! \! _{\{ 1 \}}\) the covariant Lorenz gauge fixing fermion from \colsaref{col:linearized_de_donder_gauge_fixing_fermion}{col:covariant_lorenz_gauge_fixing_fermion}, respectively. Then the complete gauge fixing and ghost Lagrange densities for (effective) Quantum General Relativity coupled to the Standard Model can be generated via \(D \Upsilon\).
\end{thm}

\begin{proof}
	This follows directly from the calculation
	\begin{equation}
	\begin{split}
		D \Upsilon & = \bigl ( P + Q \bigr ) \bigl ( \stigma^{(1)} + \digamma \! \! _{\{ 1 \}} \bigr ) \\
		& = P \stigma^{(1)} + P \digamma \! \! _{\{ 1 \}} + Q \stigma^{(1)} + Q \digamma \! \! _{\{ 1 \}} \\
		& \simeq_\text{TD} P \stigma^{(1)} + Q \digamma \! \! _{\{ 1 \}} \, ,
	\end{split}
	\end{equation}
	where we have used \lemref{lem:p_tensor_densities} and \(\simeq_\text{TD}\) means equality modulo total derivatives.
\end{proof}

\enter

\begin{lem} \label{lem:relation_BRST_anti-BRST_operators}
	The BRST operators are related to their corresponding anti-variants via ghost conjugation. In particular, we have:
	\begin{subequations}
	\begin{align}
		\overline{P} & \equiv P^{\dagger_C} \, , \\
		\overline{Q} & \equiv Q^{\dagger_c}
		\intertext{and}
		\overline{D} & \equiv D^{\dagger} \, ,
	\end{align}
	\end{subequations}
	where \(\dagger_C\), \(\dagger_c\) and \(\dagger\) denote the graviton-ghost, gauge ghost and total ghost conjugations from \defnref{defn:ghost-conjugation}, respectively.
\end{lem}

\begin{proof}
	This follows immediately from the respective definitions.
\end{proof}

\enter

\begin{thm} \label{thm:relation_cocomplex_complex}
	The BRST cocomplexes are isomorphic to the anti-BRST complexes in negative degree via ghost conjugation:
	\begin{subequations}
	\begin{align}
		\bigl ( \CQ^{i,j}, P^{i,j} \bigr )^{\dagger_C} & \cong \bigl ( \CQ_{-i}^j, \overline{P}_{-i}^j \bigr ) \, , \\
		\bigl ( \CQ^{i,j}, Q^{i,j} \bigr )^{\dagger_c} & \cong \bigl ( \CQ^i_{-j}, \overline{Q}^i_{-j} \bigr )
		\intertext{and}
		\bigl ( \CQ^k, D^k \bigr )^\dagger & \cong \bigl ( \CQ_{-k}, \overline{D}_{-k} \bigr ) \, ,
	\end{align}
	\end{subequations}
	for all \(i, j, k \in \mathbb{Z}\), where \(i\) denotes the graviton-ghost degree, \(j\) the gauge ghost degree and \(k\) the total ghost degree.
\end{thm}

\begin{proof}
	Let \(S \in \setbig{P, Q, D}\) be any of the three BRST operators and \(\overline{S} \in \setbig{\overline{P}, \overline{Q}, \overline{D}}\) be any of the three anti-BRST operators. Furthermore, we use the coordinates on the sheaf of particle fields \(\FQ\) with the shifted anti-Hermitian Lautrup--Nakanishi auxiliary fields from \defnref{defn:ghost-conjugation}, i.e.\
	\begin{subequations}
	\begin{align}
		{B^\prime}^\rho & \coloneq B^\rho - \frac{\varkappa \zeta}{2} \left ( \overline{C}^\sigma \bigl ( \partial_\sigma C^\rho \bigr ) - \bigl ( \partial_\sigma \overline{C}^\rho \bigr ) C^\sigma \right )
		\intertext{and}
		{b^\prime}^a & \coloneq b^a - \frac{\mathrm{g} \xi}{2} \tensor{f}{^a _b _c} \overline{c}^b c^c \, .
	\end{align}
	\end{subequations}
	Then, the ghost conjugations from \defnref{defn:ghost-conjugation} constitute the claimed isomorphisms between the cochain complexes \((\CQ^\bullet, S^\bullet)\) and the chain complexes \((\CQ_{- \bullet}, \overline{S}_{- \bullet})\) due to \lemref{lem:relation_BRST_anti-BRST_operators}.
\end{proof}

\enter

\begin{col} \label{col:relation_cohomology_homology}
	 In particular, the BRST cohomologies are isomorphic to the anti-BRST homologies in negative degree:
	\begin{subequations}
	\begin{align}
		H^{i,j} \bigl ( P \bigr ) & \cong H_{-i}^j \bigl ( \overline{P} \bigr ) \, , \\
		H^{i,j} \bigl ( Q \bigr ) & \cong H^i_{-j} \bigl ( \overline{Q} \bigr )
		\intertext{and}
		H^k \bigl ( D \bigr ) & \cong H_{-k} \bigl ( \overline{D} \bigr ) \, ,
	\end{align}
	\end{subequations}
	again for all \(i, j, k \in \mathbb{Z}\), where \(i\) denotes the graviton-ghost degree, \(j\) the gauge ghost degree and \(k\) the total ghost degree. In particular, the ghost conjugation operators induce a zero-dimensional Poincaré duality.
\end{col}

\begin{proof}
	This is an immediate consequence of \thmref{thm:relation_cocomplex_complex}.
\end{proof}

\enter

\begin{defn}[Sign-twisted anti-BRST operators] \label{defn:sign_twisted_anti-BRST_operators}
	Let \(\overline{S} \in \setbig{\overline{P}, \overline{Q}, \overline{D}}\) be any of the three anti-BRST operators, then we call
	\begin{equation}
		\widetilde{S}^l \coloneq \left ( -1 \right )^l \overline{S}^l
	\end{equation}
	the corresponding sign-twisted anti-BRST operator, where \(l \in \mathbb{Z}\) denotes the respective degree. Crucially, all BRST operators \(S \in \setbig{P, Q, D}\) now commute instead of anticommute with their sign-twisted anti-BRST operator \(\widetilde{S}\).
\end{defn}

\enter

\begin{prop} \label{prop:cochain_chain_homotopy}
	Given the situation of \defnref{defn:sign_twisted_anti-BRST_operators} and let \(\widetilde{\mathcal{S}} \coloneq S \circ \widetilde{S}\). Then the sign-twisted anti-BRST operators \(\widetilde{S}\) are cochain homotopies between \(\widetilde{\mathcal{S}}\) and the zero map. Similarly, the BRST operators \(S\) are chain homotopies between \(\widetilde{\mathcal{S}}\) and the zero map. In particular, the maps \(\widetilde{\mathcal{S}}\) are null homotopic for all \(\widetilde{\mathcal{S}} \in \setbig{\widetilde{\mathcal{P}}, \widetilde{\mathcal{Q}}, \widetilde{\mathcal{D}}}\).
\end{prop}

\begin{proof}
	Again, let \(S \in \setbig{P, Q, D}\) be any of the three BRST operators, with corresponding sign-twisted anti-BRST operator \(\widetilde{S}\) and their composition \(\widetilde{\mathcal{S}} \coloneq S \circ \widetilde{S}\). Then, the statements about cochain and chain homotopies are an immediate consequence of the following diagram:
	\begin{subequations}
	\begin{equation}
	\begin{tikzcd}[row sep=huge]
		\ddots \arrow{dr}{2 \widetilde{\mathcal{S}}} & \vdots \arrow{d}{\widetilde{S}} \arrow{dr}{2 \widetilde{\mathcal{S}}} & \vdots \arrow{d}{\widetilde{S}} \arrow{dr}{2 \widetilde{\mathcal{S}}} & \\
		\cdots \arrow{r}{S} \arrow{dr}{2 \widetilde{\mathcal{S}}} & \CQ^\bullet \arrow{r}{S} \arrow{d}{\widetilde{S}} \arrow{dr}{2 \widetilde{\mathcal{S}}} & \CQ^{\bullet +1} \arrow{r}{S} \arrow{d}{\widetilde{S}} \arrow{dr}{2 \widetilde{\mathcal{S}}} & \cdots \\
		\cdots \arrow{r}{S} \arrow{dr}{2 \widetilde{\mathcal{S}}} & \CQ^{\bullet -1} \arrow{r}{S} \arrow{d}{\widetilde{S}} \arrow{dr}{2 \widetilde{\mathcal{S}}} & \CQ^\bullet \arrow{r}{S} \arrow{d}{\widetilde{S}} \arrow{dr}{2 \widetilde{\mathcal{S}}} & \cdots \\
		& \vdots & \vdots & \ddots
	\end{tikzcd}
	\end{equation}
	Explicitly, we interpret the diagonal maps as
	\begin{equation}
		2 \widetilde{\mathcal{S}} \equiv 2 \widetilde{\mathcal{S}} - 0 \, ,
	\end{equation}
	\end{subequations}
	i.e.\ the composition map minus the zero map.
\end{proof}

\enter

\begin{rem}
	\propref{prop:cochain_chain_homotopy} gives an interpretation for gauge fixing fermions that are induced by gauge fixing bosons in the sense of \eqnref{eqn:gauge-fixing-boson}: This will be discussed in \cite{Prinz_6}. Furthermore, we remark that starting from \defnref{defn:sign_twisted_anti-BRST_operators} similar statements hold when sign-twisting the BRST operator \(S\) instead of the anti-BRST operator \(\overline{S}\): This can be seen easily, as the construction only depends on the additional relative minus signs, which makes them commute instead of anticommute. Moreover, we remark that the maps \(\widetilde{\mathcal{S}} \coloneq S \circ \widetilde{S}\) are the sign-twisted versions of the \emph{super-BRST operators} introduced in \cite[Definition 2.5]{Prinz_6}. Additionally, they can also be understood as the corresponding Hodge--Laplace operators of the respective sign-twisted (co)complexes.
\end{rem}

\enter

\section{Conclusion} \label{sec:conclusion}

We have studied the BRST double complex of (effective) Quantum General Relativity coupled to the Standard Model. To this end, we started with a review of the geometric underpinnings, notably graded supergeometry, in \sectionref{sec:geometric-setup}. Then, we have studied the diffeomorphism and gauge complexes separately in \sectionsaref{sec:diffeomorphism-brst-complex}{sec:gauge-brst-complex}, respectively. In particular, we have recalled that the BRST and anti-BRST operators are nilpotent and thus cohomological or homological differentials, respectively. In addition, we have discussed the gauge fixing fermions for the de Donder, linearized de Donder, Lorenz and covariant Lorenz gauge fixing conditions. A particularly important result is \lemref{lem:p_tensor_densities}, which characterizes all Lagrange densities that are essentially closed with respect to the diffeomorphism BRST operator and diffeomorphism anti-BRST operator as scalar tensor densities of weight \(w = 1\). Finally, we study the corresponding double complex in \sectionref{sec:diffeomorphism-gauge-brst-double-complex}: Our main results are that all BRST and anti-BRST operators anticommute and thus give rise to the corresponding \emph{total BRST operator} and \emph{total anti-BRST operator}, cf.\ \thmref{thm:total_brst_operator} and \colref{col:total_anti-brst_operator}. Furthermore, we have shown that graviton-ghosts decouple from matter of the Standard Model if the gauge fixing fermion of Yang--Mills theory is a tensor density of weight \(w = 1\), cf.\ \thmref{thm:no-couplings-grav-ghost-matter-sm}. Moreover, we have shown that all gauge fixing and ghost Lagrange densities of (effective) Quantum General Relativity coupled to the Standard Model can be derived from a \emph{total gauge fixing fermion}, via the action of the \emph{total BRST operator}, cf.\ \thmref{thm:total_gauge_fixing_fermion}. Additionally, we have shown in \thmref{thm:relation_cocomplex_complex} that the BRST cocomplexes are isomorphic to the anti-BRST complexes in negative degree via ghost conjugation. In particular, we have observed in \colref{col:relation_cohomology_homology} that this implies that the corresponding cohomologies are related to the respective homologies in negative degree. Finally, we have shown in \propref{prop:cochain_chain_homotopy} that the sign-twisted anti-BRST operators and the corresponding BRST operators can be interpreted as (co)chain homotopies between their respective composition and the zero map. In a direct follow-up article \cite{Prinz_6}, we have used the findings of this article to derive the symmetric (i.e.\ Hermitian) ghost Lagrange densities for (effective) Quantum General Relativity and covariant Quantum Yang--Mills theory. This was first introduced in \cite{Baulieu_Thierry-Mieg_1} for the case of pure Quantum Yang--Mills theory. Additionally, we have applied our results to study the transversal structure of (effective) Quantum General Relativity coupled to the Standard Model in \cite{Prinz_7}. Finally, we have also considered the case of perturbative Quantum Gravity with a cosmological constant in \cite{Prinz_8}. In future work, we would like to construct a \emph{perturbative BRST cocomplex} as a \emph{Feynman graph cocomplex}, as follows \cite{Prinz_9}: This will be set up such that its cohomology groups consist of transversal linear combinations of Feynman graphs. Then, for theories without gauge anomalies, we construct the corresponding renormalization Hopf algebra directly on the respective cohomology groups, to manifestly combine transversality with renormalization in \cite{Prinz_3}.

\section*{Acknowledgments}
\addcontentsline{toc}{section}{Acknowledgments}

The author thanks Peter Teichner and Chris Wendl for clarifying discussions on chain homotopies. This research is supported by the \emph{Kolleg Mathematik Physik Berlin} of the Humboldt University of Berlin and the University of Potsdam via the research group of Sylvie Paycha.

\bibliography{References}{}
\bibliographystyle{babunsrt}

\end{document}